%% file: main.tex
\documentclass{article}

\PassOptionsToPackage{numbers, compress}{natbib}
\usepackage{amsmath,amssymb,amsthm,fullpage,mathrsfs,pgf,tikz,caption,subcaption,mathtools}
\usepackage[ruled,vlined,linesnumbered]{algorithm2e}
\usepackage{natbib}
\usepackage{amsmath,amssymb,amsthm,mathtools}
\usepackage[utf8]{inputenc} 
\usepackage[T1]{fontenc}    
\usepackage{hyperref}       
\usepackage{url}            
\usepackage{booktabs}       
\usepackage{amsfonts}       
\usepackage{nicefrac}       
\usepackage{microtype}      
\usepackage[margin=1in]{geometry}
\usepackage{bbm}
\usepackage{enumitem}
\usepackage{xcolor}
\usepackage{cleveref}

\let\oldnl\nl
\newcommand{\nonl}{\renewcommand{\nl}{\let\nl\oldnl}}

\input{macro}

\title{Point Location and Active Learning: Learning Halfspaces Almost Optimally}

\author{%
  Max Hopkins\thanks{Department of Computer Science and Engineering, UCSD, California, CA 92092. Email: \texttt{nmhopkin@eng.ucsd.edu}. Supported by NSF Award DGE-1650112.}
  \and
  Daniel Kane\thanks{Department of Computer Science and Engineering / Department of Mathematics, UCSD, California, CA 92092. Email: \texttt{dakane@eng.ucsd.edu}. Supported by NSF Award CCF-1553288 (CAREER) and a Sloan
Research Fellowship.}
    \and
    Shachar Lovett\thanks{Department of Computer Science and Engineering, UCSD, California, CA 92092. Email: \texttt{slovett@cs.ucsd.edu}. Supported by NSF Award CCF-1909634.}
    \and
    Gaurav Mahajan\thanks{Department of Computer Science and Engineering, UCSD, California, CA 92092. Email: \texttt{gmahajan@eng.ucsd.edu}}
}

\begin{document}

\maketitle

\begin{abstract}
Given a finite set $X \subset \mathbb{R}^d$ and a binary linear classifier $c: \mathbb{R}^d \to \{0,1\}$, how many queries of the form $c(x)$ are required to learn the label of every point in $X$? Known as \textit{point location}, this problem has inspired over 35 years of research in the pursuit of an optimal algorithm. Building on the prior work of Kane, Lovett, and Moran (ICALP 2018), we provide the first nearly optimal solution, a randomized linear decision tree of depth $\tilde{O}(d\log(|X|))$, improving on the previous best of $\tilde{O}(d^2\log(|X|))$ from Ezra and Sharir (Discrete and Computational Geometry, 2019). As a corollary, we also provide the first nearly optimal algorithm for actively learning halfspaces in the membership query model. En route to these results, we prove a novel characterization of Barthe's Theorem (Inventiones Mathematicae, 1998) of independent interest. In particular, we show that $X$ may be transformed into approximate isotropic position if and only if there exists no $k$-dimensional subspace with more than a $k/d$-fraction of $X$, and provide  a similar characterization for exact isotropic position.
\end{abstract}
\newpage
\section{Introduction}
Consider the following combinatorial question: an arbitrary finite set of points $X \subset \mathbb{R}^d$ is divided into two parts by an arbitrary hyperplane $h \in \mathbb{R}^d$. If for any $x \in \mathbb{R}^d$ you are allowed to ask its label with respect to $h$, i.e.
\[
\sign(\langle x, h \rangle) \in \{-,0,+\},
\]
what is the minimum number of questions needed to determine the label of every point in $X$? It is easy to observe that using such ternary questions (known as linear queries), deciding between the $\Omega(|X|^d)$ possible labelings requires at least $\Omega(d\log(|X|))$ queries. A matching upper bound, on the other hand, seems to be not so simple. Indeed the search for such an algorithm has been ongoing for over 35 years~\cite{der1983polynomial,meiser1993point,cardinal2015solving,kane2018generalized,ezra2019nearly}, and for good reason: slight variants and special cases of this simple question underlie many fundamental problems in computer science and mathematics. 

This problem has been studied in the literature from two distinct, but equivalent, viewpoints. In machine learning, the question is a variant of the well studied problem of learning linear classifiers. From this standpoint, $X$ is viewed as an arbitrary set of data points, and $h$ as a hyperplane chosen by some adversary. The hyperplane defines a halfspace, which in turn acts as a binary (linear) classifier for points in $X$. The goal of the learner is to label every point in $X$ with respect to this classifier in as few linear queries as possible. In the machine learning literature, this type of learning model is referred to as active learning with membership query synthesis~\cite{angluin1988queries}, a semi-supervised branch of Valiant's~\cite{valiant1984theory} PAC-learning model.

On the other hand, in computational geometry and computer graphics, the problem is often viewed from its dual standpoint called \textit{point location}. In this view, $X$ is thought of as a set of hyperplanes in $\mathbb{R}^d$ rather than a set of points, and partitions the space $\mathbb{R}^d$ into cells. Given a point $h\in \mathbb{R}^d$, point location asks the mathematically equivalent question: how quickly can we determine in which cell $h$ lies? Work in this area has often centered around its relation to important combinatorial problems like k-SUM or KNAPSACK~\cite{der1983polynomial,kane2018near,ezra2019nearly}, which reduce to special cases of point location. Since these views are equivalent, for the remainder of this paper we will adopt the standpoint common to machine learning, but will generally still refer to the problem as point location.

Regardless of which view is taken, up until this point the best algorithm for point location on $n$ points in $\mathbb{R}^d$ took $\tilde{O}(d^2\log(n))$\footnote{We use $\tilde{O}$ to hide poly-logarithmic factors in the dimension $d$.} queries, leaving a quadratic gap between the upper and lower bounds. In this work we close this gap up to a sub-polynomial factor, providing significant progress towards resolving the decades old question regarding the tightness of the $\Omega(d\log(n))$ information theoretic lower bound. In particular, we prove the existence of a nearly optimal randomized linear decision tree (LDT) for the point location problem in two models: bounded-error ($\delta$-reliable), and zero-error (reliable). For the former, we build a randomized LDT of depth $O(d\log^2(d)\log(n/\delta))$, and require an additional additive $d^{1+o(1)}$ factor for the more difficult reliable model. Our bounded-error linear decision tree can be combined with standard results on PAC-learning halfspaces~\cite{vapnik1974theory,valiant1984theory,Blumer,hanneke2016optimal} to provide the first nearly optimal learner for this class over arbitrary distributions.

\subsection{Computational and Learning Models}\label{sec:comp-models}
In this work we study point location from the viewpoint of randomized linear decision trees (LDTs), a model common to both computational geometry and machine learning. Randomization allows us to study two different solution regimes: zero-error, and bounded-error.
\subsubsection{Linear Decision Trees}
A linear decision tree $T$ is a tree whose internal nodes represent ternary linear queries, and whose leaves represent solutions to the problem at hand. Given a hyperplane $h\in\mathbb{R}^d$, a linear query is one of the form:
\[
Q_x(h) = \sign(\langle x, h \rangle ) \in \{-,0,+\},
\]
where $x \in \mathbb{R}^d$. The outgoing edges from an internal node indicate the branch to follow based on the sign of the linear query. The leaves, in our case, stand for labelings $T(h): X \to \{-,0,+\}$ of the instance space $X$. 
To execute the linear decision tree on input $h$, we begin at the root node. At each node $v$, we test the sign of a linear query with respect to some $x \in \mathbb{R}^d$ and proceed along the corresponding outgoing edge. We do this until we reach a leaf, at which point, we output its corresponding labeling. The complexity of a linear decision tree is measured by its depth, denoted $d(T)$, the length of the longest path from the root node to a leaf. Since it will be convenient in the following definitions, we will also define the input specific depth $d(T,h)$ to be the number of internal nodes reached on input $h$. 
\subsubsection{Randomized LDTs: Zero-Error}
A randomized linear decision tree $T=(\mathcal{D},\mathcal T)$ is a distribution $\mathcal{D}$ over a (possibly infinite) set $\mathcal T$ of deterministic LDTs. In computational geometry, many works~\cite{kane2017,kane2018generalized,ezra2019nearly} on the point location problem focus mostly on randomized LDTs which label all points in $X$ without making any errors. We will say that a randomized LDT satisfying this property \textit{reliably} computes the point location problem. Formally,
\begin{definition}[Reliable computation]
A randomized decision tree $T=(\mathcal{D},\mathcal T)$ reliably solves the point location problem on a set $X \subset \mathbb{R}^d$ if for any hyperplane $h \in \mathbb{R}^d$, the following holds:
\[
\forall T' \in \mathcal T, \forall x \in X: T'(h)(x) = \sign(\langle x, h \rangle)
\]
\end{definition}
In other words, every deterministic LDT that makes up our randomized LDT must correctly solve the point location problem. The main difference between this randomized model and deterministic LDTs then lies in how we measure the complexity of the object. For a deterministic LDT, the standard measure of complexity is its depth, the number of queries which must be made in order to learn the sign of every point on the worst-case input. In a randomized LDT on the other hand, for a given $h$ there may exist a few trees in $\mathcal T$ with large input specific depth--so long as these trees are not chosen too often. More precisely, our measure of complexity will be the worst-case expected input specific depth of a tree drawn from $\mathcal{D}$.
\begin{definition}
The expected depth of a randomized decision tree $T=(\mathcal{D},\mathcal T)$, denoted $\ED(T)$ is:
\[
\text{ED}(T) = \max\limits_{h \in \mathbb{R}^d}\underset{T' \sim \mathcal{D}}{\mathbb{E}}[d(T',h)].
\]
\end{definition}
Equivalently, we can think of this parameter as the expected number of queries to learn the worst case input.
\subsubsection{Randomized LDTs: Bounded-Error}
While reliably solving point location is interesting from a computational geometry perspective and useful for building zero-error algorithms for k-SUM and related problems~\cite{kane2018near,ezra2019nearly}, for many applications such a strict requirement is not necessary. In the PAC-learning model, for instance, we are only interested in approximately learning the hidden hyperplane with respect to some adversarially chosen distribution over $\mathbb{R}^d$. Here we do not need every LDT in our distribution to return the right answer - it is sufficient if most of them do. We will say an LDT $\delta$-reliably computes (or solves) the point location problem on $X$ if with probability at least $1-\delta$ it makes no errors. Formally, 
\begin{definition}[$\delta$-reliable computation]
A randomized decision tree $T=(\mathcal{D},\mathcal T)$ $\delta$-reliably solves the point location problem on a set $X \subset \mathbb{R}^d$ if for any hyperplane $h \in \mathbb{R}^d$, the following holds:
\[
\Pr\limits_{T' \sim \mathcal{D}}\left [\exists x \in X: T'(h)(x) \neq \sign(\langle x, h \rangle) \right ] < \delta
\]
\end{definition}
Since we are now allowing a small probability of error, as is commonly the case in randomized algorithms we will change our complexity measure to be worst-case rather than average. In particular, our complexity measure for $\delta$-reliably computing the point location problem will be the maximum depth of any $T' \in \mathcal T$.

\begin{definition}
The maximum depth of a randomized decision tree $T=(\mathcal{D},\mathcal T)$, denoted $\MD(T)$ is:
\[
\text{MD}(T) = \underset{T' \in \mathcal T}{\max}\left ( d(T') \right ).
\]
\end{definition}
Using this worst-case complexity measure will help us transfer our results to an active variant of PAC-learning theory which uses a similar worst-case measure called query complexity.
\subsubsection{PAC and Active Learning}
Probably Approximately Correct (PAC) learning is a learning framework due to Valiant~\cite{valiant1984theory} and Vapnik and Chervonenkis~\cite{vapnik1974theory}. In the PAC-learning model, given a set (instance space) $X$, and a family of classifiers $\mathcal H$ (where $h \in \mathcal H$ maps elements in $X$ to a binary label in $\{0,1\}$), first an adversary chooses a distribution $D$ over $X$ and a classifier $h \in \mathcal H$. Then, the learner receives labeled samples from the distribution with the goal of approximately learning $h$ (up to $\varepsilon$ accuracy) with high  probability $(1-\delta)$ in the fewest number of samples. The complexity of learning a concept class $(X,\mathcal H)$ is given by its \textit{sample complexity}, the minimum number of samples $n(\varepsilon,\delta)$ such that there exists a learner $A$ which maps samples from $D$ to labelings which achieve this goal:
\begin{align}\label{eq:PAC-def}
\forall D,h: \Pr_{S \sim D^{n(\varepsilon,\delta)}}\left [ \Pr_{x \sim D}\left[A(S)(x) \neq h(x)\right] \leq \varepsilon \right ] \geq 1-\delta.
\end{align}
Active learning, on the other hand, more closely mirrors today's challenges with big data where samples are cheap, but labels are expensive. Using the same adversarial setup, active learning\footnote{The model of active learning we use is referred to as the Pool + Membership Query Synthesis (MQS) model}~\cite{angluin1988queries} provides the learner with unlabeled samples from $D$ along with an oracle which the learner can query to learn the label $h(x)$ for any $x\in X$. The goal of active learning is to adaptively query the most informative examples,  and thus exponentially decrease the required number of labeled points over Valiant's initial ``passive'' model. In this case, the complexity measure of learning $(X,\mathcal H)$ is given by its \textit{query complexity} $q(\varepsilon,\delta)$, the minimum number of queries (oracle calls) required to achieve the same guarantee from \Cref{eq:PAC-def}.
\subsubsection{Non-homogeneous Hyperplanes and Binary Queries}
In the previous sections, we assumed that both the labeling of $X$ and the queries making up our LDT are ternary and given by a homogeneous (through the origin) hyperplane. While these assumptions are standard in the point location literature, in learning theory it is more common to consider binary queries, and, when possible, the more general class of non-homogeneous hyperplanes. In this generalized scenario, the labels of $X$ and linear queries of our decision tree instead take the form:
\[
Q_x(h,b)=\text{sign}(\langle x, h \rangle + b) \in \{-,+\}.
\]
While we assume throughout most of the body (Sections \ref{sec:proof-overview}-\ref{sec:verification}) that we are in the ternary, homogeneous case, we show in \Cref{sec:non-homogeneous} how to extend our results to the binary, non-homogeneous case.
\subsection{Results}\label{sec:results}
We prove the existence of nearly optimal LDTs (learners) for all three of these regimes. For the bounded-error regime, our result is only a $\text{poly}\log(d)$ factor away from being optimal.
\begin{theorem}\label{intro:bounded-LDT}
Let $X \subset \mathbb{R}^d$ be an arbitrary finite set. Then there exists a randomized LDT $T$ that $\delta$-reliably computes the point location problem on $X$ with maximum depth:
\[
\MD(T) \leq O\left(d\log^2(d)\log\left(\frac{n}{\delta}\right)\right)
\]
\end{theorem}
Combining this with standard PAC-learning results for classes with finite VC-dimension~\cite{Blumer}, we prove the existence of the first nearly optimal active learner for halfspaces (either homogeneous or non-homogeneous), denoted $\mathcal H_d$, over arbitrary distributions.
\begin{corollary}\label{intro:active}
The query complexity of actively learning $(\mathbb{R}^d,\mathcal H_d)$ over any distribution, with membership queries is:
\[
q(\varepsilon,\delta) \leq O\left(d\log^2(d)\log\left(\frac{d}{\varepsilon\delta}\right)\right )
\]
\end{corollary}
By adding a verification step to our LDT that $\delta$-reliably computes point location, we prove the existence of an LDT that reliably computes point location at the cost of an additive $d^{1+o(1)}$ factor.
\begin{theorem}\label{intro:zero-LDT}
Let $X \subset \mathbb{R}^d$ be an arbitrary finite set. Then there exists a randomized LDT $T$ that  reliably computes the point location problem on $X$ with expected depth:
\[
\ED(T) \leq O\left(d\log^2(d)\log(n)\right) + d \cdot 2^{O\left(\sqrt{\log(d)\log\log(d)}\right)}.
\]
\end{theorem}
For large enough $n$, this bound is optimal up to a polylogarithmic factor in $d$. 
\subsection{Related Work}\label{sec:related-work}
\subsubsection{Point Location}\label{sec:point location}
The early roots of the point location problem as we study it\footnote{In its most general form, the point location problem asks about any type of partition, and reaches back to the mid 1970's or earlier~\cite{dobkin1976multidimensional}.} stem from the study of other classic combinatorial problems. As such, Meyer auf Der Heide's~\cite{der1983polynomial} early work giving an $O(n^4\log(n))$ linear search algorithm for the $n$-dimensional Knapsack problem is often considered the seminal work for this area. Meiser~\cite{meiser1993point} later stated the problem in the form we consider, and provided an 
$O(d^5\log(n))$ depth linear decision tree for a general set of $n$ hyperplanes in $\R^d$. These results were later improved by Cardinal, Iacono, and Ooms~\cite{cardinal2015solving}, Kane, Lovett, and Moran~\cite{kane2018generalized,kane2018near} and finally by Ezra and Sharir~\cite{ezra2019nearly} to an expected depth of $\tilde{O}(d^2\log(n))$. 

Our work expands on the margin-based technique introduced by Kane, Lovett, and Moran in~\cite{kane2018generalized}, which by itself gives an $\tilde{O}(d^3\log(n))$ expected depth LDT. Kane, Lovett, and Moran use the fact that point location is invariant to invertible linear transformations of the data to transform $X$ into isotropic position. It is then possible to take advantage of the structure introduced by this transformation to apply a margin-based technique from their earlier work with Zhang~\cite{kane2017}. We employ a similar overall strategy, repeatedly transforming the remaining unlabeled points into isotropic position, but by novel structural analysis and a new inference technique based upon dimensionality reduction, we improve to a nearly tight $\tilde{O}(d\log(n))+d^{1+o(1)}$ expected depth randomized LDT.
\subsubsection{Learning Halfspaces}
Learning halfspaces is one of the oldest and most studied problems in learning theory. We cover here only a small fraction of works, those which are either seminal or closely related to our results. The first such work is the classic of Blumer, Ehrenfeucht, Haussler, and Warmuth~\cite{blumer1989learnability}, who showed using VC theory~\cite{vapnik1974theory} that $d$-dimensional halfspaces can be PAC-learned in only $O((d\log(1/\varepsilon)+\log(1/\delta))/\varepsilon)$ labelled examples, and that this bound is nearly tight (much later, Hanneke~\cite{hanneke2016optimal} tightened the result). In the years since, many works showed that active learning could in some cases exponentially decrease the number of labelled examples required to learn. Many of these works focused on learning halfspaces in the Pool-based model (where in contrast to MQS, only sampled points may be queried) with restricted distributions. A series of papers~\cite{balcan2007margin,dasgupta2006coarse,freund1997selective} showed increasingly improved bounds for learning homogeneous halfspaces over (nearly) uniform distributions on the unit ball. Later this work was extended to more general classes of distributions~\cite{balcan2013active,balcan2017sconcave}, finally giving a nearly optimal $\tilde O(d \log(1/\varepsilon))$ algorithm. These results were then extended to the non-homogeneous case in a more powerful query model that allowed the learner to compare points~\cite{kane2017,hopkins2019aid}.

Work that focused on learning over adversarial distributions, on the other hand, tended to use the stronger MQS learning model. The most efficient theoretical algorithms in this regime are (implicit) results from the point location literature, most recently Ezra and Sharir's~\cite{ezra2019nearly} result translates to a roughly $\tilde{O}(d^2\log(1/\varepsilon))$ query algorithm under adversarial distributions. On the practical end, a number of works~\cite{chen2017near, alabdulmohsin2015efficient} presented MQS algorithms that seemed experimentally to achieve the $d\log(1/\varepsilon)$ lower bound, but none could do so provably. Our work is the first to provably match the practical performance of these heuristic methods.
\paragraph{Paper organization.}
In \Cref{sec:proof-overview} we provide a high-level overview of our proof. 
In \Cref{sec:oracle} we describe the margin oracle, which is a new concept central in this work.
In \Cref{sec:weakLearner}, we cover our core weak learner. In \Cref{sec:boosting}, we prove the existence of an LDT that $\delta$-reliably computes point location, and show how it provides a nearly optimal active learner of homogeneous halfspaces. In \Cref{sec:verification}, we show how to efficiently verify our learners in order to build an LDT that reliably computes point location. In \Cref{sec:non-homogeneous} we show how to extend our results to non-homogeneous halfspaces and binary queries.

\section{Proof Overview}\label{sec:proof-overview}
We begin with a proof sketch that highlights our overall strategy. Our overview itself will be split into four main sections. In \Cref{overview:approach}, we provide a high-level sketch of our approach as a whole. In \Cref{sec:overview-ww}, we examine how to build an efficient bounded-error weak learner that is at the core of both \Cref{intro:bounded-LDT} and \Cref{intro:zero-LDT}. In \Cref{sec:overview-boosting}, we show how to boost this weak learner to obtain \Cref{intro:bounded-LDT} and \Cref{intro:active}. In \Cref{sec:overview-verification}, we show how to boost and verify the weak learner of \Cref{sec:overview-ww} to obtain the zero-error LDT of \Cref{intro:zero-LDT}. 

In order to provide a simple explanation of our strategy, we leave out many technical details which we cover in the following sections. 
\subsection{Overall Approach} \label{overview:approach}
To start, we briefly provide some intuition for our general approach. We will assume for simplicity throughout Sections \ref{sec:proof-overview}-\ref{sec:verification} that we are in the homogeneous case, and show in \Cref{sec:non-homogeneous} how our arguments generalize to the non-homogeneous case. Recall then the setup: we are given a finite set $X \subset \R^d$, there is an unknown homogeneous hyperplane (given by a normal vector) $h \in \R^d$, and our goal is to label $X$ using linear queries of the form $\sign(\langle x, h \rangle)$, where we can use any $x \in \R^d$.

Assume we are given some point of reference $x_{\text{ref}} \in \R^d$ and an orthonormal basis $\{v_1,\ldots,v_d\}$ for $\R^d$, and that it is possible to learn for each $i$ up to high accuracy $\frac{\langle v_i, h \rangle}{\langle x_{\text{ref}}, h \rangle}$. Since each point $x \in X$ can be written in terms of this orthonormal basis, this would allow us to estimate $\frac{\langle x, h \rangle}{\langle x_{\text{ref}}, h \rangle}$ up to high accuracy, and thus learn the label of $x$. This strategy alone will work to label all of $X$ efficiently, unless it contains many points with small margin (inner product with $h$). 

Kane, Lovett, and Moran~\cite{kane2018generalized}, using a different inference strategy, ran into this same fundamental issue. They circumvent the problem (for sets in general position) via two key facts: point location is invariant to invertible linear transformations, and there exists such a transformation on $X$ that ensures many points have large margin. Unfortunately, if the remaining points have very small margin, then this method loses polynomial factors in the dimension $d$. Concretely, the result of~\cite{kane2018generalized} required $\tilde{O}(d^3 \log n)$ queries.

We solve this problem, and achieve a near-linear dependence on the dimension $d$, by using these small margin points to apply dimensionality reduction. In essence, we argue that for every hyperplane $h$ and a ``nice'' set of points $X \subset \R^d$, there exists a parameter $1 \le k \le d$ for which the following holds. It is possible to split $\R^d$ into two orthogonal subspaces: one of dimension $d-O(k)$ with \q{low margin}, and another of dimension $O(k)$ with \q{high margin}. In addition, a $k/d$ fraction of the points in $X$ have large projection to the high margin subspace. This allows us to significantly reduce the dimension of the problem (from $d$ to $k$), and allows our learner to label large margin points in $\tilde{O}(k)$ queries rather than $\tilde{O}(d)$ queries. In essence, we couple the fraction of points that we label to the number of queries we make.

Formally, the \q{nice} structure we need is an approximate isotropic position. Building on previous works on \q{vector scaling} problems~\cite{barthe1998reverse,forster2002linear,dvir2014breaking}, we give an exact characterization of when a point set $X$ can be transformed to such a configuration (See \Cref{prop:dense-subspace} for details). This allows us to extend our analysis to an arbitrary finite set of points. In summary, this procedure allows us to infer a $k/d$ fraction of an arbitrary point set $X$ using only $\tilde{O}(k)$ queries. \Cref{intro:bounded-LDT} and \Cref{intro:zero-LDT} both follow from applying different boosting procedures to this process. The former uses a weighting scheme between iterations to ensure every point is learned with high probability, while the latter runs a slightly more costly verification process on each learner to ensure that no mistakes have been made.  
\subsection{Weak Learner}\label{sec:overview-ww}
Both of our main results, then, are based upon boosting a highly efficient weak learner: a randomized algorithm $A$ that makes linear queries and returns (abusing notation) a partial labeling $A: X \to \{-,0,+,\bot\}$ in which $\bot$ is interpreted as ``don't know.'' Before introducing our core weak learner, we need to introduce some further notation and terminology.
\begin{definition}
Let $X \subset \mathbb{R}^d$ be a finite set of unit vectors, and $\mu$ a distribution over $X$. We denote such a pair by $(X,\mu)$. In cases where the ambient dimension is not clear from context, we denote it by $(X \subset \R^d, \mu)$.
\end{definition}
Note that assuming $X$ consists of unit vectors does not lose any generality for point location or active learning homogeneous hyperplanes, since normalizing points does not change their label. For notational simplicity, we often refer to a fraction of $X$ with respect to $\mu$ simply as a fraction of $(X,\mu)$. For instance, it is useful when dealing with weak learners to be able to talk about the fraction of $(X,\mu)$ they label. We call this value their \textit{coverage}.
\begin{definition}[Coverage]
Given a pair $(X, \mu)$ and $A$ a weak learner, $A$'s coverage with respect to $(X,\mu)$, $C_{\mu}(A)$, is a random variable (over the internal randomness of $A$) denoting the measure of $X$ learned:
\begin{align*}
    C_{\mu}(A) = \Pr_{x \sim \mu}[A(x) \neq \bot].
\end{align*}
\end{definition}
It will also be useful to have a way to talk about the correctness of a weak learner. For this, we adopt notation introduced in~\cite{hopkins2020noise}:
\begin{definition}[Reliability]
Given a pair $(X,\mu)$, we say that a learning algorithm $A$ is $p$-reliable if with probability $1-p$ over the internal randomness of $A$, the labeling output by $A$ makes no errors. If $A$ never makes an error, we call in reliable.
\end{definition}
These definitions are all we need to present our core weak learner, \textsc{PartialLearn}:
\begin{theorem}[Informal \Cref{thm:arbitrary-ww-learner}: \textsc{PartialLearn}]\label{overview:arbitrary-ww-learner}
Given a pair $(X \subset \R^d,\mu)$ and $p>0$, there exists a $p$-reliable weak learner \textsc{PartialLearn} with the following guarantees. For any (unknown) hyperplane $h \in \R^d$, with probability $1-p$ there exists $\frac{1}{10} \le k \le d$ such that:
\begin{enumerate}
    \item \textsc{PartialLearn} has coverage at least $k/d$.
    \item \textsc{PartialLearn} makes at most $O(k\log(d)\log(d/p))$ queries.
\end{enumerate}
\end{theorem}
\subsubsection{Isotropic Position}
We will start by proving an intermediary result, \textsc{IsoLearn}: a weak learner for a pairs $(X,\mu)$ which satisfy a structural property called $\varepsilon$-\textit{isotropic position}.
\begin{definition}[$\varepsilon$-isotropic Position]
A pair $(X\subset \mathbb{R}^d,\mu)$ lies in $\varepsilon$-isotropic position if:
\[
    \forall v \in \mathbb{R}^d:
    (1-\varepsilon)\frac{1}{d} \leq \sum\limits_{x \in X} \mu(x)\frac{\langle x, v \rangle^2}{\norm{v}^2} \leq (1+\varepsilon)\frac{1}{d}
\]
\end{definition}
The key to our learner's efficiency lies in its ability to exploit this structure to find a parameter $k$ such that it can infer the labels of a $\frac{k}{d}$ fraction of $(X,\mu)$, while using only $\tilde{O}(k)$ queries.
\begin{lemma}[Informal \Cref{lemma:isotropic-oracle-learner}: \textsc{IsoLearn}]\label{overview:ww-learner}
Let the pair $(X \subset \R^d,\mu)$ be in $1/4$-isotropic position, and let $p>0$. There exists a $p$-reliable weak learner \textsc{IsoLearn} with the following guarantees. For any (unknown) hyperplane $h\in \mathbb{R}^d$, with probability $1-p$ there exists $\frac{1}{10} \leq k \leq d$ such that:
\begin{enumerate}
    \item \textsc{IsoLearn} has coverage at least $k/d$.
    \item \textsc{IsoLearn} makes at most $O(k\log(d)\log(d/p))$ queries.
\end{enumerate}
\end{lemma}
Throughout the rest of this section, we assume the (unknown) hyperplane is non-zero (this situation may occur in our algorithm if points in $X$ lie on the hyperplane $h$). We may assume this without loss of generality as it is easy to test up front by checking if the sign of a few random points is zero. We will break down the construction of \textsc{IsoLearn} into three parts: finding the structure in $X$, dimensionality reduction, and inference.
\subsubsection{Finding the Structure in $X$} The structure at the core of \textsc{IsoLearn} is one common to the machine learning literature, the concept of \textit{margin}.
\begin{definition}[Margin]
Given a hyperplane $h\in\mathbb{R}^d$, the margin of a point $x \in \mathbb{R}^d$ is its inner product with $h$, $\langle x, h \rangle$. Since $h$ is not restricted to be unit length, we will often work with the \textbf{normalized margin} of $x$, which we denote by $m(x,h)$:
\[
m(x,h) = \frac{\langle x, h \rangle}{\norm{h}}.
\]
\end{definition}
Margin-based algorithms are common in both the active learning and point location literature (e.g.~\cite{balcan2013active,kane2017,kane2018generalized}). In their recent work on point location, Kane, Lovett, and Moran~\cite{kane2018generalized} observed that sets in $\varepsilon$-isotropic position must have an $\Omega(1/d)$ fraction of points with normalized margin $\Omega(1/\sqrt{d})$. This follows from the fact that the sum of the squared normalized margins of points in $X$ is at least $(1-\varepsilon)\frac{n}{d}$ -- if not enough points have large margin, their squared sum cannot be this large. Through the cleaner inference technique presented in this work, this observation alone is enough to build a randomized LDT with expected depth $\tilde{O}(d^2\log(n))$. Reaching near-linear depth, however, requires insight into the finer-grain structure of $X$. The idea is to show $X$ has a \q{gap} in margin, that is parameters $t$ and $s$ such that not too many points have normalized margin between $t$ and $t/s$. \textsc{IsoLearn} will work by exploiting this gap, ignoring in a sense points with margin less than $t/s$ in order to learn the fraction with margin greater than $t$. We formalize this key structure in the following definition.
\begin{definition}[$(k,t,s,c)$-structured]
We call a pair $(X \subset \R^d,\mu)$ $(k,t,s,c)$-structured with respect to a hyperplane $h \in \R^d \setminus \{0\}$ if it satisfies the following properties:
\begin{enumerate}
    \item Many points in $X$ have normalized margin at least $t$:
    \begin{align}\label{cond:large-mar}
        \Pr_{x \sim \mu}\left[|m(x,h)| \geq t \right] \geq \frac{k}{d}
    \end{align}
    \item Many points in $X$ have normalized margin at most $t/s$:
    \begin{align}\label{cond:small-mar}
        \Pr_{x \sim \mu}\left[|m(x,h)| \leq t/s \right] \geq 1-\frac{ck}{d}
    \end{align}
\end{enumerate}
When clear from context, we often drop the phrase ``with respect to $h$''. Similarly, throughout the paper we will use the shorthand $(k,t,s)$-structured for $(k,t,s,5)$-structured.
\end{definition}
\textsc{IsoLearn} relies on the following structural insight concerning pairs $(X,\mu)$ in $1/4$-isotropic position:
\begin{lemma}[\Cref{lemma:k-t-structure}]\label{overview:lemma-k-t-struct}
Let the pair $(X \subset \R^d,\mu)$ be in $1/4$-isotropic position. Then for any hyperplane $h\in \mathbb{R}^d \setminus \{0\}$ and $s>2$, there exist parameters $k$ and $t$ satisfying:
\begin{enumerate}
    \item $1 \leq k \leq d$.
    \item $t \geq s^{-O(\log(d))}$.
    \item $(X,\mu)$ is $(k,t,s)$-structured.
\end{enumerate}
\end{lemma}
In \Cref{sec:find-struct}, we prove an algorithmic variant of this result (\Cref{lemma:t-k}). In particular, we show that with high probability it is possible to find parameters $k$ and $t$ such that $X$ is $(k,t,s)$-structured\footnote{As we will note later in greater detail, $t$ here is found implicitly in terms of the margin of some reference point $x_{\text{ref}} \in \mathbb{R}^d$.}.

\subsubsection{Dimensionality Reduction}
For the moment, assume we have found some parameters $k \le d$, $t \ge d^{-O(\log(d))}$, $s=d^{\Omega(1)}$, such that $X$ is $(k,t,s)$-structured. 
When $k = \Omega(d)$, \Cref{cond:large-mar} implies that a constant fraction of $X$ has normalized margin at least $t$, so no dimensionality reduction is required. Thus, we focus on the case where $k \ll d$, where our goal will be to reduce the dimension of the problem from $d$ to $k$.

Our basic strategy will be to decompose $\mathbb{R}^d$ into two orthogonal subspaces: $V$, a high dimensional but \q{small margin} subspace, and $V^{\perp}$, a low dimensional but \q{large margin} subspace. The structural result we need is the following.


\begin{lemma}[\Cref{lemma:subspace-structure}]\label{overview:subspace-structure}
Let $(X \subset \R^d,\mu)$ be $(k,t,s)$-structured with respect to a hyperplane $h \in \mathbb{R}^d \setminus \{0\}$ where $k<d/10$. Then there exists a subspace $V$ with the following properties:
\begin{enumerate}
    \item $V$ is high dimensional: \[\text{Dim}(V) = d-O(k)\]
    \item $V$ has small margin with respect to $t$: 
    \[
    \forall v \in V, |m(v,h)| \leq O\left( \norm{v}d\frac{t}{s}\right)
    \]
\end{enumerate}
\end{lemma}
The subspace $V$ in this case is approximately the span of points in $X$ with margin less than $t/s$. In \Cref{sec:dim-reduce}, we prove an algorithmic version of this result (\Cref{lemma:subspace-algorithm}) that finds an analogous subspace $V$ with high probability by using a random sample of such small margin points. For simplicity, we denote below $e=\dim(V^{\perp})$ where $e=O(k)$. Further, we set $s=d^{\Omega(1)}$ to guarantee a large enough gap between \q{large margin} and \q{small margin}.
  
\subsubsection{Inference}\label{overview:sec-inference}
For any subspace $V$, we can write any $x \in X$ in terms of an orthonormal basis $v_i$ for $V$, and $w_i$ for $V^{\perp}$:
\begin{align*}
x = \sum\limits_{i=1}^{d-e} \alpha_i v_i+ \sum\limits_{i=1}^{e} \beta_i w_i
\end{align*}
for some set of constants $-1 \leq \alpha_i,\beta_i \leq 1$ (recall that $x$ is a unit vector). Since the inner product is bi-linear, this means we can express the margin of $x$ through the margins of $v_i$ and $w_i$:
\begin{align}\label{overview:eq-expand-basis}
\langle x,h \rangle= \sum\limits_{i=1}^{d-e} \alpha_i \langle v_i, h \rangle+ \sum\limits_{i=1}^{e} \beta_i \langle w_i, h \rangle.
\end{align}
The idea behind finding a high-dimensional subspace $V$ with the properties given in \Cref{overview:subspace-structure} is that the lefthand term, $x$'s projection onto $V$, does not have much effect on $x$'s margin. This is because the basis vectors $v_i$ are unit length, and thus satisfy a small margin condition guaranteeing that:
\[
m(x,h) \in \left (\sum\limits_{i=1}^{e} \beta_i m(w_i,h) \right ) \pm \frac{t}{d^{\Omega(1)}}.
\]
Inferring the sign of $x$ then reduces to learning information about the smaller dimensional space $V^\perp$. In particular, imagine that we could learn up to an additive error of $\frac{t}{d^{\Omega(1)}}$ the normalized margin of each $w_i$. Calling this value $\gamma_i$, we would be able to express the normalized margin of $h$ as:
\begin{align*}
m(x,h) &\in \left ( \sum\limits_{i=1}^{e} \beta_i \gamma_i \right) \pm \frac{t}{d^{\Omega(1)}}.
\end{align*}
Recall that $(X,\mu)$ is $(k,t,s)$-structured, meaning at least a $k/d$ measure of points have normalized margin at least $t$. For such a point $x$, notice that most of its margin must come from the lefthand sum. Further, as long as this sum is at least $\frac{t}{d^{\Omega(1)}}$, we can infer the sign of $x$. This process would allow us to infer the sign of any point with margin at least $t$.

Unfortunately, it is not clear that it is possible to learn the margin of each $w_i$. However, given a reference point $x_{\text{ref}}$ whose normalized margin was $t$, then by querying the sign of the point $w_i-\alpha x_{\text{ref}}$, we could check relations of the form:
\[
m(w_i,h) \overset{?}{\geq} \alpha t.
\]
Note that there exists some $|\alpha_w| \leq t^{-1}$ such that
\[
m(w_i,h) = \alpha_w t,
\]
since the normalized margin of $w_i$ is at most $1$. Finding $\alpha_w$ up to an additive error of $\frac{1}{d^{\Omega(1)}}$ is then sufficient for our inference. This can be done by binary searching over $\alpha$ in only $\log\left(\frac{t^{-1}}{d^{\Omega(1)}}\right) = O(\log^2(d))$ queries. Since there are only $e=O(k)$ of these vectors, we only need a total of $O(k\log^2(d))$ queries in total. 

As a final note, we can use this same inference method for the case where $k \geq \frac{d}{10}$, but do not require the dimensionality reduction step -- $V^\perp$ will just be all of $\mathbb{R}^d$.

In \Cref{sec:oracle} we cover in more detail the mechanism behind the reference point $x_{\text{ref}}$, found via taking a random combination of points in our set. This technique turns out to be crucial both for identifying $(k,t,s)$-structure, and for efficiently finding the subspace $V$. These processes are covered in \Cref{sec:weakLearner}, where we show how both may be done in no more than $O(k\log(d)\log(d/p))$ queries, which completes the proof.

\subsubsection{Arbitrary Sets and Distributions}\label{sec:overview-arb-ww-learner}
\textsc{IsoLearn} is only an intermediary result since most pairs $(X,\mu)$ are not in $1/4$-isotropic position. However, drawing on the results of Barthe~\cite{barthe1998reverse}, Forster~\cite{forster2002linear}, and  Dvir, Saraf, and Wigderson~\cite{dvir2014breaking}, we can prove a related structural result: any pair $(X,\mu)$ contains a subspace dense in $X$ with respect to $\mu$ which may be transformed into $1/4$-isotropic position.
\begin{proposition}[\Cref{cor:iso}]\label{overview:prop-iso}
Given a pair $(X \subset \R^d,\mu)$, for some $1 \leq k \leq d$, for all $\varepsilon > 0$ there exist:
\begin{enumerate}
    \item A $k$-dimensional subspace $V$ with the property $\mu(X \cap V) \geq \frac{k}{d}$ and
    \item An invertible linear transformation $T:V \to V$ such that the pair $((X \cap V)_T,(\mu|_{X \cap V})_T)$:
\[
(X \cap V)_T = \left \{ x_T = \frac{T(x)}{\norm{T(x)}}: x \in X \cap V \right\},\qquad
(\mu|_{X \cap V})_T = \frac{\mu(x)}{\mu(X \cap V)}
\]
is in $\varepsilon$-isotropic position.
\end{enumerate}
\end{proposition}
The proof of \Cref{overview:arbitrary-ww-learner} follows from \Cref{overview:prop-iso} via an observation of Kane, Lovett, and Moran~\cite{kane2018generalized}: point location is invariant to invertible linear transformations. In greater detail, let $h' = (T^{-1})^{\top}(h)$, and $x' = \frac{T(x)}{\norm{T(x)}}$. Kane, Lovett, and Moran observe that $\langle x',h' \rangle = \langle \frac{x}{\norm{T(x)}}, h \rangle$. Thus not only is it sufficient to learn the labels of $x'$ with respect to $h'$, but we can simulate linear queries on $h'$ simply by normalizing $x$ by an appropriate constant. Applying \Cref{overview:prop-iso}, we can then apply Barthe's isotropic transformation~\cite{barthe1998reverse} on some dense subspace, and apply \textsc{IsoLearn} to complete the proof. The details are covered in \Cref{sec:arb-ww-learner}.
\subsection{Bounded-Error: Boosting}\label{sec:overview-boosting}
Now we will show how to boost this weak learner into an LDT that $\delta$-reliably solves the point location problem.
\begin{theorem}[\Cref{thm:bldt}]\label{overview:bldt}
Let $X \subset \mathbb{R}^d$, $|X|=n$. Then there exists a randomized LDT $T$ that $\delta$-reliably computes the point location problem on $X$ with maximum depth:
\[
\MD(T) \leq O\left(d\log^2(d)\log\left(\frac{n}{\delta}\right)\right).
\]
\end{theorem}
To make this boosting process simpler, we will start by applying the boosting process from~\cite{kane2018generalized} to build a stronger weak learner with constant coverage.
\begin{lemma}[Informal \Cref{lemma:weakLearner}: \textsc{WeakLearn}]\label{overview:w-leaner}
Let $X \subseteq \mathbb{R}^d$ be a set and $\mu$ a distribution over $X$. There exists a weak learner \textsc{WeakLearn} with the following guarantees. For any (unknown) hyperplane $h\in \mathbb{R}^d$, the following holds:
\begin{enumerate}
    \item \textsc{WeakLearn} is $.01$-reliable
    \item With probability at least $.99$, \textsc{WeakLearn} has coverage at least $.99$.
    \item \textsc{WeakLearn} uses at most $O(d\log^2(d))$ queries
\end{enumerate}
\end{lemma}
The strategy for proving \Cref{overview:w-leaner} is simple. At each step, we restrict to the set of un-inferred points and run another instance of \textsc{PartialLearn} with probability parameter $p=1/\text{poly}(d)$. Since each instance learns a $k/d$ measure of points in $O(k\log(d)\log(d/p))$ queries for some $k$, repeating this process until we have used $O(d\log^2(d))$ queries is sufficient (see \Cref{sec:boosting} for details). 

Unfortunately, continuing this strategy to learn all of $X$ would force us to set our probability parameter $p$ to be inverse polynomial in $\log(n)$, costing an additional $\log\log(n)$ factor in the depth of our LDT. Instead, we will employ the fact that \textsc{WeakLearn} can learn over any distribution to apply a boosting process that re-weights $X$ at each step. The idea is that by multiplicatively reducing the weight on points learned by a certain iteration, we can ensure that, with high probability, every point is labeled in at least $5\%$ of the learners. Since each learner is correct with $99\%$ probability, a Chernoff bound tells us that the majority label for each point will then be correct with probability at least $1-\delta$ as desired. We give the details of this process in \Cref{sec:boosting}.
\subsection{Zero-error: Verification}\label{sec:overview-verification}
\textsc{PartialLearn} (\Cref{overview:arbitrary-ww-learner}) comes with a small probability of error. In this section, we explain our strategy for verifying the weak learner to build a randomized LDT that reliably computes point location.
\begin{theorem}[\Cref{thm:zldt}]\label{overview:zldt}
Let $X \subset \mathbb{R}^d$ be an arbitrary finite set. Then there exists a randomized LDT $T$ that  reliably computes the point location problem on $X$ with expected depth:
\[
\ED(T) \leq O\left(d\log^2(d)\log(n)\right) + d \cdot 2^{O\left(\sqrt{\log(d)\log\log(d)}\right)}.
\]
\end{theorem}
The core of this theorem lies in showing how verifying the  labelings which stem from naively boosting \textsc{PartialLearn} (in the sense of \Cref{overview:w-leaner}) reduces to a related combinatorial problem we term \textit{matrix verification}:
\begin{definition}[Matrix Verification]
Let $S \subset \mathbb{R}^d$ be a subset of size $m$, $h \in \mathbb{R}^d$ a hyperplane, and $\{C_{ij}\}_{i,j=1}^m$ a constraint matrix in $\R^{m \times m}$. We call the problem of determining whether for all $i,j$:
\[
\langle  x_i, h \rangle \leq  C_{ij}\langle x_j, h \rangle
\]
a \textbf{matrix verification problem} of size $m$. Further, we denote by $V(m)$ the minimum expected number of queries made across randomized algorithms which solve verification problems of size $m$ in any dimension.
\end{definition}
In fact, we show a two-way equivalence between point location and matrix verification: without too much overhead, point location reduces to a small matrix verification problem, and matrix verification reduces in turn to solving several point location problems in fewer dimensions. This recursive structure allows us to efficiently solve both problems. 

In this section, we sketch both directions of this equivalence and show how to use them to build the zero-error LDT from \Cref{intro:zero-LDT}. To start, we examine how verifying point location reduces to matrix verification. To do so, we must first examine the source of errors in the weak learners we wish to verify.
\subsubsection{The Source of Errors} 
We have not yet covered exactly where the error is introduced into \textsc{PartialLearn} (\Cref{overview:arbitrary-ww-learner}). The culprit is the \q{low margin} subspace $V$, which with some small probability, may not actually satisfy the key assumption:
\begin{align}\label{overview:eq-basis}
\forall v \in V: | \langle v, h \rangle | \leq \frac{\norm{v}}{d^{\Omega(1)}} |\langle  x_{\text{ref}}, h \rangle|.
\end{align}
In more detail, we build $V$ via the covariance matrix of a sample $S$ of $O(d\log(d))$ small margin points from $X$. Our weak learner checks probabilistically that points in $S$ satisfy:
\begin{align}\label{overview:eq-small}
\forall s \in S: | \langle s, h \rangle | \leq \frac{1}{d^{\Omega(1)}} |\langle  x_{\text{ref}}, h \rangle|,
\end{align}
which if true, verifies \Cref{overview:eq-basis} (See \Cref{lemma:subspace-structure}). The source of the error is then tied intrinsically to this set of equations for $S$. If we can verify that these equations hold, we can assure that we have not mislabeled any points.
\subsubsection{Learning Relative Margins}
We would like to show that verifying \Cref{overview:eq-small} over boosted instances of \textsc{PartialLearn} reduces to a small matrix verification problem. The key idea in this reduction lies in noticing that the sets of equations for each instance are related. This stems from the fact that our inferences in \textsc{PartialLearn} come with extra information: their relative margin to $x_{\text{ref}}$. 

Recall our inference strategy from the proof sketch of \textsc{IsoLearn}. Writing a point $x \in \mathbb{R}^d$ in terms of orthonormal bases $v_i$ for $V$ and $w_i$ for $V^\perp$, we use our knowledge of their relation to $x_{\text{ref}}$ to bound $x$'s normalized margin:
\begin{align*}
\frac{\langle x, h \rangle}{\norm{h}} &\in \frac{\langle x_{\text{ref}}, h \rangle}{\norm{h}}\left ( \sum\limits_{i=1}^{e} \gamma_i\beta_i \pm \frac{1}{d^{\Omega(1)}}\right).
\end{align*}
In \Cref{overview:sec-inference}, we used this to show that as long as the sum $\left | \sum\limits_{i=1}^{e} \gamma_i\beta_i \right |$ is large enough, we can infer the sign of $x$. However, the equation actually provides additional information that was not useful until this point: $\sum\limits_{i=1}^{e} \gamma_i\beta_i$ also acts as an approximation of the relative margin of $x$ to $x_{\text{ref}}$. 
\subsubsection{From Point Location to Matrix Verification}
It turns out that this approximation is all we need to reduce to matrix verification. Assume for a moment that $|X|=\text{poly}(d)$, and boost \textsc{PartialLearn} naively until all of $X$ is labeled. Since $X$ is small, this only takes $O(d\log(d))$ rounds. For the $i$th round, let $S_i$ denote the sample $S$ of small points, and $x_i$ denote the reference point $x_{\text{ref}}$. For each round, our  goal is to verify the $O(d\log(d))$ equations from \Cref{overview:eq-small}.

To reduce this problem to matrix verification, we will start by verifying the final weak learner (which we can do directly). Assume now, inductively, that we have verified equations for all but the first $i$ learners. Since all of $X$ is labeled in this process and $S_i$ is a sample from $X$, each $s_i \in S_i$ must be labeled by some future weak learner. By the previous observation, this further means that the relative margin of each $s_i \in S_i$ is approximately known to one of the reference points $x_j$ for $j>i$. To verify \Cref{overview:eq-small}, it is then sufficient to check a set inequalities of the form:
\begin{align}
    \langle x_i, h \rangle \leq C_{ij}\langle x_j, h \rangle
\end{align}
for some constant $C_{ij}$ (see \Cref{sec:verification} for details). These equations then form a matrix verification problem of size $m=O(d\log(d))$ on the reference points $x_i$. If we set the probability parameter $p$ of our weak learner to be $1/\text{poly}(d)$, then we can also be assured that verification will succeed with at least $50\%$ probability, allowing us to bound the expected number of queries for point location on any $X \subset \mathbb{R}^d$ of size $n=\text{poly}(d)$, denoted $T(n,d)$, by:
\begin{align}\label{overview:p-to-v}
T(d^{O(1)},d) \leq C_1d\log^3(d) + 2V(C_2d\log(d)),
\end{align}
for some constants $C_1,C_2$. The full details are covered in \Cref{lemma:pl-to-veri}.
\subsubsection{From Matrix Verification to Point Location}
Solving this matrix verification problem naively, however, still involves checking $\tilde{O}(d^2)$ equations. To escape this, observe that for each row $i$, there is a single inequality that is necessary and sufficient to verify the rest: the minimum over $j$ of $C_{ij}\langle x_j, h \rangle$. If we can efficiently compute the minimum index for each row, we can solve the problem in only $\tilde{O}(d)$ queries. 

The key observation of this section is that finding this minimum can be rephrased as a point location problem. In particular, we can compute the minimum by knowing for all distinct $i,j,k$:
\[
\sign \left(C_{ij}\langle x_{j}, h \rangle - C_{ik}\langle x_{k}, h \rangle \right) = \sign \left(\langle C_{ij}x_{j} - C_{ik} x_{k}, h \rangle \right),
\]
which is a point location problem on $\text{poly}(d)$ points in $d$ dimensions. By a similar, divide and conquer argument, we can reduce finding the minimum to several point location problems in \textit{fewer} than $d$ dimensions. Consider dividing the constraint matrix $C$ into batches of columns with indices $C_1 = \{1,\ldots,b\}, C_2 = \{b+1,\ldots, 2b\}$,... and finding the minimum index just within each subset. By directly comparing these minima, we can stitch together a global solution. Further, each sub-problem corresponds to a point location problem in $b$ dimensions, since its point set lies in the span of at most $b$ points (e.g. $\{x_1,\ldots,x_b\}$). This implies that we can bound the expected number of queries to solve matrix verification by:
\begin{align}\label{overview:v-to-p}
    V(m) \leq 2m + \frac{m^2}{b}+ \frac{2m}{b}  T(mb^2,b).
\end{align}
The full details are covered in \Cref{lemma:ver-to-pl}.
\subsubsection{Putting it all Together}
Combining Equations~\eqref{overview:p-to-v} and \eqref{overview:v-to-p} sets up a recurrence that implies the following bound on $T(n,d)$ for $n=\text{poly}(d)$:
\begin{align}\label{overview:final-p-bound}
    T(d^{O(1)},d) \leq d \cdot 2^{O\left(\sqrt{\log(d)\log\log(d)}\right)}.
\end{align}
However, since we are interested in arbitrarily large $X$, we must make one final adjustment. Instead of initially running the entire boosting process before verification, we will verify batches of $d^2$ learners at a time. While we cannot directly apply the above process to a batch of weak learners (since as described, the process requires learning \textit{all} of $X$ to work), we can turn the verification of the batch into a $d$ dimensional point location problem on \text{poly}$(d)$ points. The expected number of queries to verify a batch is then given by \Cref{overview:final-p-bound}, and combining this with the query bounds on our weak learner from \Cref{overview:arbitrary-ww-learner} proves that the overall process gives a randomized LDT satisfying the conditions of \Cref{intro:zero-LDT}. The full details are covered in \Cref{cor:veri-time} and \Cref{thm:zldt}.
\section{The Margin Oracle}\label{sec:oracle}
In the previous sections we have left aside much of the technical insight of our work in favor of providing a simple overview of our proof techniques. Here, we highlight an important technique which underlies many of our results. First, recall from \Cref{sec:overview-ww} the concept of \textit{margin}.
\begin{definition}[Margin]
Given a hyperplane $h\in\mathbb{R}^d$, the margin of a point $x \in \mathbb{R}^d$ is its inner product with $h$, $\langle x, h \rangle$. Since $h$ is not restricted to be unit length, we will often work with the \textbf{normalized margin} of $x$, $m(x,h)$:
\[
m(x,h) = \frac{\langle x, h \rangle}{\norm{h}}.
\]
\end{definition}
\subsection{The Margin Norm}
Previous margin-based strategies~\cite{kane2017,kane2018generalized} for point location revolved around the ability to compare the margin of two points $a,b \in \mathbb{R}^d$. While useful, this technique alone will be too inefficient for our purposes. In our work, given finite subsets $A,B \subset \mathbb{R}^d$, we will often need to approximately compare the margins of points in $A$ to those in $B$ using only a single query. This technique, for instance, is integral to building the \q{small margin} subspace $V$ we need for dimensionality reduction. To this effect, given a subset $S \subset \mathbb{R}^d$, we will define a single value called the margin norm that encompasses the information we need. 

\begin{definition}
We define the margin norm of a finite subset $S \subset \mathbb{R}^d$ with respect to a hyperplane $h \in \mathbb{R}^d$ as:
\[
\norm{S}_h \coloneqq \sqrt{\sum\limits_{x \in S} m(x,h)^2}
\]
\end{definition}
Notice that if we define $\vec{S}_h=(m(x,h):x \in S)$ to be the vector of normalized margins of points in $S$, then the margin norm is just the L2-norm of $\vec{S}_h$. The margin norm provides us with a single value we can compare between subsets, but we still need some way of comparing these values via linear queries.
\subsection{The Margin Oracle}
With this in mind, we introduce the margin oracle. Given a subset $S$, the margin oracle produces a ``representative'' point $x_S$ that has normalized margin approximately $\norm{S}_h$. 
\begin{definition}
Let $h \in \mathbb{R}^d$ be a hyperplane, $S \subset \mathbb{R}^d$ a finite set, and $\lambda\geq1$ some constant. The margin oracle $\bigo_\lambda(S)$ outputs a $\lambda$-representative point $x_S \in \mathbb{R}^d$ in the sense that:
\[
\lambda^{-1} \norm{S}_h \leq |m(x_S,h)| \leq \lambda\norm{S}_h.
\]
\end{definition}
It turns out that access to a margin oracle is sufficient to build \textsc{IsoLearn} (\Cref{overview:ww-learner}), and thus to nearly optimally reliably compute point location. Unfortunately, LDTs do not generally come equipped with such an oracle, so we need a way to simulate it through more standard methods.
\subsection{Simulating the Margin Oracle using Gaussian Combinations}
In fact, the margin oracle is remarkably easy to simulate probabilistically. Our ability to do so stems from a simple observation: the normalized margin of a random Gaussian combination of elements in $S$ is itself a Gaussian variable whose variance is $\norm{S}_h^2$.
Let $S=\{x_1,\ldots,x_{|S|}\}$.
We define the Gaussian combination $x_S$ as:
\[
x_S = \sum\limits_{i=1}^{|S|} x_ig_i,
\]
where each $g_i$ is independently drawn from a standard Gaussian. The margin of $x_S$, which we denote $X_S$, is a Gaussian whose variance is the squared margin norm of $S$:
\[
X_S = \left \langle \sum\limits_{i=1}^{|S|} x_ig_i, \frac{h}{\norm{h}} \right \rangle \sim \mathcal N \left (0, \norm{S}^2_h \right ).
\]
By standard Gaussian concentration and anti-concentration bounds, $x_S$ will satisfy the requirements of an output from $\bigo_\lambda(S)$ with probability at least $1-O(\lambda^{-1})$. For completeness, we include the proof here.
\begin{lemma}\label{lemma:simulate-margin-oracle}
Let $S \subseteq \mathbb{R}^d$, $h\in \mathbb{R}^d$ be a hyperplane, and set $\lambda \geq 5$. We can simulate $\bigo_\lambda(S)$ with probability at least $1-\frac{5}{\lambda}$.
\end{lemma}

\begin{proof}
We sample $x_S = \sum x_i g_i$, whose normalized margin $X_S$ satisfies
$X_S \sim \mathcal N\left(0, \norm{S}^2_h\right)$.
Let $\sigma = \norm{S}_h$. We bound the range that $X_S$ lies in with high probability via standard concentration and anti-concentration bounds for Gaussian variables. In particular, a variable $X_S$ drawn from $\mathcal N(0,\sigma^2)$ satisfies:
\begin{align*}
    Pr[|X_S| \geq \lambda\sigma] \leq 2e^{-\frac{\lambda^2}{2}} \leq \frac{5}{2\lambda}
\end{align*}
and 
\begin{align*}
    Pr[|X_S| \leq \lambda^{-1} \sigma ] = \text{erf}\left(\frac{1}{\sqrt{2}\lambda}\right) \leq \frac{5}{2\lambda}
\end{align*}
\end{proof}

\section{Weak Learner}\label{sec:weakLearner}
Now that we have covered our basic proof structure and techniques, we focus on building the core weak learner, \textsc{PartialLearn}. Recall from \Cref{sec:overview-ww} that a weak learner is a randomized algorithm $A$ that makes linear queries and returns (abusing notation) a partial labeling $A: X \to \{-,0,+,\bot\}$ in which $\bot$ is interpreted as ``don't know''.
Before introducing the learner itself, we recall as well some notation and terminology.
\begin{definition}
Let $X \subset \mathbb{R}^d$ be a finite set of unit vectors, and $\mu$ a distribution over $X$. We denote such a pair by $(X,\mu)$. In cases where the ambient dimension is not clear from context, we denote it by $(X \subset \R^d, \mu)$.
\end{definition}
Note that assuming $X$ consists of unit vectors does not lose any generality for point location or active learning homogeneous hyperplanes, since normalizing points does not change their label. For notational simplicity, we often refer to a fraction of $X$ with respect to $\mu$ simply as a fraction of $(X,\mu)$. For instance, it is useful when dealing with weak learners to be able to talk about the fraction of $(X,\mu)$ they label. We call this value their \textit{coverage}.
\begin{definition}[Coverage]
Given a pair $(X, \mu)$ and $A$ a weak learner, $A$'s coverage with respect to $(X,\mu)$, $C_{\mu}(A)$, is a random variable (over the internal randomness of $A$) denoting the measure of $X$ learned:
\begin{align*}
    C_{\mu}(A) = \Pr_{x \sim \mu}[A(x) \neq \bot].
\end{align*}
\end{definition}
We will also need a way to talk about the correctness and overall coverage of a weak learner. For this, we adapt notation introduced in~\cite{hopkins2020noise}:
\begin{definition}[Reliability]
Given a pair $(X,\mu)$, we say that a learning algorithm $A$ is $p$-reliable if with probability $1-p$ over the internal randomness of $A$, the labeling output by $A$ makes no errors. If $A$ never makes an error, we call it reliable.
\end{definition}
With these definitions out of the way, we present \textsc{PartialLearn}, our core weak learner:
\begin{theorem}[\textsc{PartialLearn} (\Cref{overview:arbitrary-ww-learner})\label{thm:arbitrary-ww-learner}]
Given a pair $(X \subset \R^d,\mu)$ and $p>0$, there exists a weak learner \textsc{PartialLearn} with the following guarantees. For any (unknown) hyperplane $h \in \R^d$, with probability $1-p$ there exists $\frac{1}{10} \leq k \leq d$ such that:
\begin{enumerate}
    \item \textsc{PartialLearn} has coverage at least $k/d$
    \item \textsc{PartialLearn} makes fewer than $O(k\log(d)\log(d/p))$ queries.
\end{enumerate}
Further, \textsc{PartialLearn} satisfies the following global guarantees:
\begin{enumerate}
    \item \textsc{PartialLearn} is $p$-reliable.
    \item \textsc{PartialLearn} makes at most $O(d\log(d)\log(d/p))$ queries.
\end{enumerate}
\end{theorem}
We break the construction and correctness of \textsc{PartialLearn} into two sections. First, we present \textsc{IsoLearn}, a weak learner for pairs $(X, \mu)$ which satisfy a certain structural condition called approximate isotropic position.
\begin{definition}[$\varepsilon$-isotropic Position]
A pair $(X\subset \mathbb{R}^d,\mu)$ lies in $\varepsilon$-isotropic position if:
\[
    \forall v \in \mathbb{R}^d:
    (1-\varepsilon)\frac{1}{d} \leq \sum\limits_{x \in X} \mu(x)\frac{\langle x, v \rangle^2}{\norm{v}^2} \leq (1+\varepsilon)\frac{1}{d}
\]
\end{definition}
Second, by proving a novel characterization of when such a pair can be transformed into approximate isotropic position based upon the version of Barthe's theorem~\cite{barthe1998reverse} presented in~\cite{dvir2014breaking}, we show how \textsc{IsoLearn} can be used as a subroutine to obtain \textsc{PartialLearn}.
\subsection{Weak Learner for $X$ in $1/4$-Isotropic Position}
In this section we build \textsc{IsoLearn}, an intermediary weak learner for $(X,\mu)$ in $1/4$-isotropic position. 
We deviate slightly from the learner stated in our proof overview (\Cref{overview:ww-learner}), which will follow as a corollary. Here we allow our learner access to a margin oracle $\bigo_{\lambda}(\cdot)$. As a result, this version of the weak learner is fully reliable (i.e. it never makes an error).

\begin{lemma}[\textsc{IsoLearn} (\Cref{overview:ww-learner})]\label{lemma:isotropic-oracle-learner}
Let the pair $(X \subset \R^d,\mu)$ be in $1/4$-isotropic position, and let $p>0$. There exists a weak learner \textsc{IsoLearn} with access to a margin oracle $\bigo_{\lambda}(\cdot)$ with the following guarantees. For any (unknown) hyperplane $h\in \mathbb{R}^d$, with probability $1-p$ there exists $\frac{1}{10} \leq k \leq d$ such that:
\begin{enumerate}
    \item \textsc{IsoLearn} has coverage at least $k/d$.
    \item \textsc{IsoLearn} makes fewer than $O(k\log(d)\log(d\lambda/p))$ queries.
\end{enumerate}
Further, \textsc{IsoLearn} satisfies the following global guarantees:
\begin{enumerate}
    \item \textsc{IsoLearn} is reliable.
    \item \textsc{IsoLearn} makes at most $O(d\log(d)\log(d\lambda/p))$ queries.
    \item \textsc{IsoLearn} makes at most $O(d\log(d)\log(d/p))$ calls to $\bigo_{\lambda}(\cdot)$.
\end{enumerate}
\end{lemma}
\textsc{IsoLearn} works by exploiting certain margin-based structure in pairs $(X, \mu)$ that lie in approximate isotropic position. Since normalized margin is ill-defined when $h$ is the zero vector, we will assume for the moment that $h$ is non-zero and deal separately with this degenerate case in the proof of \Cref{lemma:isotropic-oracle-learner}. Recall from \Cref{sec:overview-ww} that a pair $(X,\mu)$ is $(k,t,s)$-structured if there exists a gap in normalized margin between parameters $t$ and $t/s$:
\begin{definition}[$(k,t,s,c)$-structured]
We call a pair $(X \subset \R^d,\mu)$ $(k,t,s,c)$-structured with respect to a hyperplane $h \in \R^d \setminus \{0\}$ if it satisfies the following properties:
\begin{enumerate}
    \item Many points in $X$ have normalized margin at least $t$:
    \begin{align*}
        \Pr_{x \sim \mu}\left[ |m(x,h)| \geq t \right] \geq \frac{k}{d}
    \end{align*}
    \item Many points in $X$ have normalized margin at most $t/s$:
    \begin{align*}
        \Pr_{x \sim \mu}\left[ |m(x,h)| \leq t/s \right] \geq 1-\frac{ck}{d}
    \end{align*}
\end{enumerate}
When clear from context, we often drop the phrase ``with respect to $h$''. Similarly, throughout the paper we will use the shorthand $(k,t,s)$-structured for $(k,t,s,5)$-structured.
\end{definition}
We divide our proof into three sections for clarity. First, we focus on finding the $(k,t,s)$-structure in $(X,\mu)$. Second, we discuss how to employ this structure to build the high-dimensional \q{low margin} subspace $V$ used for dimensionality reduction. Finally, we close out the proof by covering how we infer the labels of a $\frac{k}{d}$ fraction of $(X,\mu)$.
\subsubsection{Finding Structure in $(X,\mu)$}\label{sec:find-struct}
Before discussing the algorithmic aspect, we first prove that that any $(X,\mu)$ in $1/4$-isotropic will always be $(k,t,s)$-structured. 
\begin{lemma}[\Cref{overview:lemma-k-t-struct}]\label{lemma:k-t-structure}
Let the pair $(X \subset \R^d,\mu)$ be in $1/4$-isotropic position. Then for any hyperplane $h\in \mathbb{R}^d \setminus \{0\}$ and $s>2$, there exist parameters $k$ and $t$ satisfying:
\begin{enumerate}
    \item $1 \leq k \leq d$.
    \item $t \geq s^{-O(\log(d))}$.
    \item $(X,\mu)$ is $(k,t,s)$-structured.
\end{enumerate}
\end{lemma}
\begin{proof}
Starting from samples of size $m$, for $m=\Theta(d)$ some sufficiently large power of $2$, we examine the median of the maximum normalized margin of points from a sample of size $m/2^i$, $0 \leq i \leq \log(m)$. For any positive integer $i$, let $S_i$ be a sample of $m/2^i$ points drawn independently from $\mu$, and let $X_i$ be the random variable whose value is the maximum normalized margin of a point in $S_i$. We denote by $M_i$ the median of $X_i$. 

To begin, we note an important property of our first median, $M_0$:
\begin{claim}
As long as $m=\Theta(d)$ is sufficiently large, $M_0 \geq \Omega\left(\frac{1}{\sqrt{d}}\right)$.
\end{claim}
\begin{proof}
Note that because $(X,\mu)$ is in $1/4$-isotropic position, the following holds:
\begin{align}\label{eq:many-large-margin}
\sum\limits_{x \in X} \mu(x)m(x,h)^2 \geq \frac{3}{4d}.
\end{align}
In particular, this implies at least a $3/(8d)$ fraction of $(X,\mu)$ has normalized margin at least $\sqrt{3/(8d)}$. Otherwise:
\begin{align*}
\sum\limits_{x \in X} \mu(x)m(x,h)^2 \leq \frac{3}{8d} \cdot 1 + \left(1 -\frac{3}{8d}\right) \cdot \frac{3}{8d} < \frac{3}{4d},
\end{align*} 
which contradicts \Cref{eq:many-large-margin}. The probability that a sample of size $m$ contains a point with normalized margin $\Omega(1/\sqrt{d})$ is then at least
$1-(1-3/(8d))^m$, which is greater than $1/2$ for $m=\Theta(d)$ sufficiently large.
\end{proof}
From here, we divide our analysis into two cases dependent on whether the median maximum margin drops by $s$ at any step.
\paragraph{Margin Gap:} Assume that for some $i \in \{0,1,2,3,\ldots,\log(m)-1\}$, there is a gap in medians, that is:
\[
M_{i+1} < \frac{M_{i}}{s}.
\]
Let $i^*$ denote the smallest $i$ such that this occurs. Then it must be the case that:
\begin{enumerate}
    \item At least half of samples of size $m/2^{i^*}$ have a point with normalized margin $M_i$, implying:
    \[
        \Pr_{x \sim \mu}\left [|m(x,h)| \geq M_{i^*} \right] \geq \Omega(2^{i^*}/d)
    \]
    \item At least half of samples of size $m/2^{i^*+1}$ have no point with normalized margin greater than $M_{i^*+1}$, implying: 
    \[
        \Pr_{x \sim \mu}\left [|m(x,h)| \leq M_{i^*}/s \right] \geq 1-O(2^{i^*}/d)    
    \]
    \item $M_{i^*} > s^{-O(\log(d))}$,
\end{enumerate}
where this final fact follows from the minimality of $i^*$ and the fact that $i$ only takes on $O(\log(d))$ possible values. Together these imply that $(X,\mu)$ is $(k,t,s)$ structured for $k=\Omega(2^{i^*})$ and $t = M_{i^*}$.
\paragraph{Large Margin:} On the other hand, assume that no such gap exists. Since $M_0 \geq \Omega\left(\frac{1}{\sqrt{d}}\right)$, we get $M_{\log(m)} \geq s^{-O(\log(d))}$. Further, since $M_{\log(m)}$ is the median of samples of one point, at least half of $(X,\mu)$ must have margin $M_{\log(m)}$, implying that $(X,\mu)$ is $(k,t,s)$-structured for $k=d/2$ and $t = M_{\log(m)}$.
\end{proof}
Having $(k,t,s)$-structure is only useful, however, if we can identify it. \Cref{alg:find-gap}, \textsc{StructureSearch}, shows a query efficient procedure that does so with high probability.
At a high level, the algorithm simply follows the strategy laid out in \Cref{lemma:k-t-structure}. For each sample size $m/2^i$, we draw a batch of samples with the goal of identifying a point whose normalized margin is approximately $M_i$. To do this, we call the margin oracle to return a representative point for each sample in the batch, and find the median of these samples using linear queries. Arguing that this empirical median approximates the true median, we can output the desired representative and a corresponding parameter $k$.

\begin{algorithm}[h!]
\SetAlgoLined
\KwIn{Pair $(X \subset \R^d, \mu)$, Margin oracle $\bigo_\lambda$, gap parameter $\ell$, and probability parameter $p$.}
\nonl \textbf{Output:} Reference point $x_{ref}$, parameter $k$.\\
\nonl \textbf{Constants:} Sample size $m= 10d$, and number of subsets $r = O\left(\log\left(\frac{\log(d))}{p}\right)\right)$.\\
$i \gets 0$\\
\While{$2^i < 2d$}  {
Draw $r$ subsets $\{S^{(i)}_j\}_{j = 1}^r$ from $\mu$ of size $|S^{(i)}_j| = \frac{m}{2^i}$.\\
Call margin oracle to receive representative points $x^{(i)}_j=\bigo_{\lambda}(S^{(i)}_j)$.\\
Run deterministic selection algorithm~\cite{blum1973time}\footnotemark on $|m(\hat x^{(i)}_j,h)|$ to find the median element $\hat x^{(i)}$.\\
    $i \gets i+1$
}
 $i \gets 0$\\
\While{$2^i < d$}  {
  \If{$|m(\hat{x}^{(i)},h)| > \ell |m(\hat{x}^{(i+1)},h)|$} 
  {
    \Return{$\mg$ $\left(x_{\text{ref}}=\hat{x}^{(i)}, k=\frac{2^{i+1}}{5}\right)$}
  }
}
\Return{$\lm$ $\left(x_{\text{ref}}=\hat{x}^{(i)}, k=\frac{2^{i+1}}{5}\right)$}
 \caption{\textsc{StructureSearch}}
 \label{alg:find-gap}
\end{algorithm}
\footnotetext{Note that the deterministic median selection of~\cite{blum1973time} requires only comparisons, which can be implemented via linear queries of the form $\sign(\langle x_1 \pm x_2,h \rangle)$, $\sign(\langle x_1, h \rangle)$, and $\sign(\langle x_2, h \rangle)$.}
\begin{lemma}[$\textsc{StructureSearch}$]\label{lemma:t-k}
Let the pair $(X \subset \R^d,\mu)$ be in $1/4$-isotropic position. Then for all $p>0,\ell>2$, \textsc{StructureSearch} has the following guarantees. For any (unknown) hyperplane $h\in \mathbb{R}^d \setminus \{0\}$, \textsc{StructureSearch} returns a parameter $k$ and point $x_{\text{ref}} \in \mathbb{R}^d$ satisfying:
\begin{enumerate}
    \item $1 \leq k \leq d$.
    \item $|m(x_{\text{ref}},h)| \geq \frac{\ell^{-O(\log(d))}}{\lambda}$.
    \item $(X,\mu)$ is $\left(k,t,s\right)$-structured with probability at least $1-p$, for $t = O\left(\frac{|m(x_{\text{ref}},h)|}{\lambda\sqrt{d}}\right)$ and $s = O\left(\frac{\ell}{\lambda^2\sqrt{d}}\right)$.
\end{enumerate}
Further, $\textsc{StructureSearch}$ satisfies the following query guarantees:
\begin{enumerate}
    \item $\textsc{StructureSearch}$ makes at most $O(\log(d)\log(d/p))$ queries.
    \item $\textsc{StructureSearch}$ makes at most $O(\log(d)\log(d/p))$ calls to $\bigo_{\lambda}(\cdot)$.
\end{enumerate}
\end{lemma}
We break the proof of this Lemma into three main claims. First, we argue that an approximate median of the margin norm is sufficient to employ a similar strategy to the proof of \Cref{lemma:k-t-structure}.
\begin{claim}\label{claim:u-median}
Let $z \ge 1$ be an integer.
Let $X_z$ be the random variable whose value is the margin norm of a random sample $S \sim \mu^z$, namely, a random sample of $z$ points drawn independently from $\mu$. Suppose that $M$ is a $u$-median of $X_z$, that is a value such that:
\begin{align}
     \Pr \left[X_z \geq M\right] &\geq u, \label{eq:u-median1}\\
    \Pr \left[X_z \leq M\right] &\geq u\label{eq:u-median2}.
\end{align}
Then:
\begin{align*}
\Pr_{x \sim \mu}\left[ |m(x,h)| \geq \frac{M}{\sqrt{z}}\right] &\geq \frac{u}{z}~\text{and}~\Pr_{x \sim \mu}\left[ |m(x,h)| \leq M\right]\geq 1-\frac{\ln(1/u)}{z}.
\end{align*}
\end{claim}
\begin{proof}
If the margin norm of a sample $S$ is at least $M$, there must exist an element in $S$ with normalized margin at least $M/\sqrt{|S|}$. This implies the following relation:
\begin{equation*}
    \Pr \left[X_z \geq M\right] \leq \Pr_{S\sim \mu^z}\left[ \exists x\in S\colon |m(x,h)| \geq \frac{M}{\sqrt{z}}\right]  \leq z \Pr_{x \sim \mu}\left[  m(x,h) \geq \frac{M}{\sqrt{z}}\right].
\end{equation*} 
 Together with \Cref{eq:u-median1}, this proves the first claim. Similarly, since the margin norm of a set $S$ being $\leq M$ implies that all points in $S$ have normalized margin $\leq M$, we get that:
 \begin{equation*}
    \Pr\left[X_z \leq M\right] \leq \Pr_{S\sim \mu^z}\left[ \forall x\in S\colon |m(x,h)| \leq M\right] = \left(\Pr_{x \sim \mu}\left[  m(x,h) \leq M\right]\right)^z.
\end{equation*} 
Together with \Cref{eq:u-median2}, this proves the second claim.
\end{proof}
Second, we note that the median of our empirical sample is indeed (with high probability) a $u$-median.
\begin{claim}\label{claim:x-median}
Let $\hat M^{(i)}$ be the empirical median of margin norms $\norm{S^{(i)}_j}_h$. With probability $1 - O(p/\log(d))$, $\hat M^{(i)}$ is a $2/5$-median of $X_z$ for $z = m/2^i$.
\end{claim}
\begin{proof}
This follows from a standard Chernoff bound.
\end{proof}
Our median oracle point will thus be within a $\lambda$ factor of a $2/5$-median with high probability. Finally, to show that $t$ will be sufficiently large, analogous to \Cref{lemma:k-t-structure} we argue that $|m(\hat x^{(0)},h)|$ is large with high probability.
\begin{claim}\label{claim:x-large}
The normalized margin of $\hat{x}^{(0)}$ is large with high probability:
\begin{align}\label{eq:x0}
\Pr_A\left[|m(\hat{x}^{(0)},h)| \geq \Omega\left ( \frac{1}{\lambda \sqrt{d}} \right)\right] \geq 1-O(p/\log(d)).
\end{align}
\end{claim}
\begin{proof}
This follows from reasoning similar to \Cref{lemma:k-t-structure}, \Cref{eq:many-large-margin}. In particular, we are guaranteed that a sample of size $m$ contains a point with normalized margin at least $\Omega(1/\sqrt{d})$ with at least some constant probability greater than $1/2$. Assuming a sample indeed contains such a point, its margin norm is bounded by $\Omega(1/\sqrt{d})$ as well. The probability that the median margin norm is at least $\Omega(1/\sqrt{d})$ is then at least $1-O(p/\log(d))$ by a standard Chernoff bound, and recalling that the margin oracle $\lambda$-approximates the margin norm gives the desired result.
\end{proof}
The proof of \Cref{lemma:t-k} follows easily from these claims.
\begin{proof}[Proof of \Cref{lemma:t-k}]
Let $\mg(i)$ denote the event that \textsc{StructureSearch} returns $\mg$ at iteration $i$, and, abusing notation slightly, let $\lm$ denote the event that it returns $\lm$. Further, for each iteration $i$, let the parameters $t_i$ and $k_i$ be:
\begin{align*}
    t_{i}=\frac{|m(\hat x^{(i)},h)|}{\lambda \sqrt{m}}~\text{and}~k_i = \frac{2^{i+1}}{5}.
\end{align*}
For notational simplicity, we also define parameters:
\[
s = \frac{\ell}{\lambda^2\sqrt{m}}~\text{and}~t^*=\frac{\ell^{-O(\log(d))}}{\lambda^2\sqrt{m}}.
\]
We wish to analyze the success probability of our algorithm. To do this, note that we need only analyze the probability of a few failure cases:
\begin{enumerate}
    \item $\mg(i) \bigwedge \left ( (X,\mu) ~\text{is not}~(k_{i},t_{i}, s)\text{-structured} \bigvee t_i < t^* \right )$
    \item $\lm  \bigwedge \left ( (X,\mu) ~\text{is not} ~(k_{i_{lm}},t_{i_{lm}}, s)\text{-structured} \bigvee t_{i_{lm}} < t^* \right )$,
\end{enumerate}
where $i_{lm}$ denotes the first $i$ such that $2^i \geq d$. Notice that unless one of these cases occurs, the algorithm successfully returns the desired parameters: $x_{\text{ref}}=\hat{x}^{(i)}$ and $k=k_i$. It is sufficient then to prove that each of these $O(\log(d))$ events occurs with probability at most $O(p/\log(d))$. To see this, we split into cases analogous to \Cref{lemma:k-t-structure}.
\paragraph{Margin Gap:} When the algorithm returns $\mg$ at iteration $i$, we observe a significant drop in margins between $\hat{x}^{(i)}$ and $\hat{x}^{(i+1)}$:
\begin{equation}
\label{eq:drop}
       |m(\hat x^{(i)},h)| \geq \ell |m(\hat x^{(i+1)},h)|
\end{equation} 
Combining \Cref{eq:drop} with \Cref{claim:u-median} and \Cref{claim:x-median} implies:
\begin{equation}
       \Pr_A[(\mg(i) \land (X,\mu) ~\text{is not}~(k_{i},t_{i}, s)\text{-structured})] \leq O(p/\log(d))
    \end{equation}
Further, since $\mg(i)$ implies $i$ is the first iteration that sees a jump, \Cref{claim:x-large} implies that $|m(\hat x^{(i)},h)| < \frac{\ell^{-O(\log(d))}}{\lambda}$ with probability at most $O(p/\log(d))$ as well.
\paragraph{Large Margin:} Notice that since $5k_{i_{lm}} \geq d$, the small margin condition (\Cref{cond:small-mar}) of $(k,t,s)$-structure is trivially satisfied. Then once again combining \Cref{claim:u-median} with \Cref{claim:x-median} implies:
    \begin{equation}
       \Pr_A[(\lm \land (X,\mu) ~\text{is not}~(k_{i_{lm}},t_{i_{lm}}, s)\text{-structured})] \leq O(p/\log(d)).
    \end{equation}
By the same reasoning as MarginGap, we have as well that $|m(\hat x^{(i_{lm})}, h)| < \frac{\ell^{-O(\log(d))}}{\lambda}$ with probability at most $O(p/\log(d))$, which completes the result.
\end{proof}
\subsubsection{Dimensionality Reduction}
\label{sec:dim-reduce}
Dependent upon the parameter $k$ returned by \textsc{StructureSearch}, our argument splits into two cases:
\paragraph{Case 1: $k \geq \frac{d}{10}$.} Assume that \textsc{StructureSearch} does not fail, and thus that $(X,\mu)$ is $\left(k,t,s\right)$-structured for $s=O\left(\frac{\ell}{\lambda^2\sqrt{d}}\right)$. When $k$ is this large, \Cref{cond:large-mar} implies that $\frac{1}{10}$ of $(X,\mu)$ has normalized margin at least $t$. In this case no dimensionality reduction is required, since we can afford to use $\tilde{O}(d)$ queries when our coverage is constant.
\paragraph{Case 2: $k < \frac{d}{10}$.} On the other hand, when $k$ is small, since only a $k/d$ fraction of points have large normalized margin, we cannot afford to use so many queries. In this case, we will instead show how to apply a dimensionality reduction technique, building an $O(k)$ dimension subspace that accounts for most of the margin of points in $X$. 

In particular, assume that $k < \frac{d}{10}$, and $(X,\mu)$ is $(k,t,s)$-structured for $t=O\left ( \frac{|m(x_{\text{ref}},h)|}{\lambda \sqrt{d}}\right)$ and $s=O\left(\frac{\ell}{\lambda^2\sqrt{d}}\right)$ as given by \textsc{StructureSearch}. Based upon this fact, we show how to construct a high-dimension \q{small margin} subspace whose orthogonal complement is the aforementioned $O(k)$ dimensional subspace. While our argument centers around the fact that this subspace is algorithmically constructable, we begin for intuition by proving not only that such a subspace exists, but perhaps more importantly that it can be built based upon the set of small margin points in $X$.
\begin{lemma}[\Cref{overview:subspace-structure}]\label{lemma:subspace-structure}
Let $(X \subset \R^d,\mu)$ be $(k,t,s)$-structured with respect to a hyperplane $h \in \mathbb{R}^d \setminus \{0\}$ where $k<d/10$. Then there exists a subspace $V$ with the following properties:
\begin{enumerate}
    \item $V$ is high dimensional: \[\text{Dim}(V) = d-O(k)\]
    \item $V$ has small margin with respect to $t$: 
    \[
    \forall v \in V, |m(v,h)| \leq O\left( \norm{v}d\frac{t}{s}\right)
    \]
\end{enumerate}
\end{lemma}
\begin{proof}
Let $S \subset X$ denote points with small normalized margin:
\begin{align*}
    S &= \left \{ x \in X: |m(x,h)| \leq \frac{t}{s} \right \}
\end{align*}
We consider the covariance matrix of $S$:
\begin{align*}
\text{Cov}(S) = \mathbb{E}_{x \sim \mu|_{S}}[xx^\top],
\end{align*}
where $\mu|_S$ is the normalized restriction of $\mu$ to $S$.
We claim that the subspace spanned by eigenvectors of $\text{Cov}(S)$ with sufficiently large eigenvalues has the desired properties. We prove this in two steps. First, we show that the span of eigenvectors of $\text{Cov}(S)$ with eigenvalues greater than $\Omega(1/d)$ satisfy property 2. Second, we show that there are $d-O(k)$ such eigenvectors.

Let $M$ denote the $d \times |S|$ matrix whose column vectors are elements in $S$, and $D$ the $|S| \times |S|$ diagonal matrix with entries $\mu|_S(s)$ for each $s \in S$. We may then equivalently write $\text{Cov}(S)$ as $MDM^\top$. Let $\{v_1,\ldots,v_m\}$ be orthonormal eigenvectors with corresponding eigenvalues $\{\lambda_1,\ldots,\lambda_m\}$, where each $\lambda_i \geq L = \Omega(1/d)$. Denote the span of $\{v_1,\ldots,v_m\}$ by $V$. To show that $v \in V$ has small normalized margin, we write $v$ as a sum over elements in $S$:
\begin{align*}
    v &= \sum\limits_{i=1}^m \alpha_i v_i\\
    &= M \left (\sum\limits_{i=1}^m \frac{\alpha_i}{\lambda_i}D M^\top v_i \right ).
\end{align*}
Since the normalized margin of elements in $S$ is bounded by $t/s$, bounding the margin of $v$  reduces to bounding the 1-norm of the right-hand coefficient vector:
\begin{align*}
    |m(v,h)| \leq \frac{t}{s} \left| \left |\left (\sum\limits_{i=1}^m \frac{\alpha_i}{\lambda_i}D M^\top v_i \right ) \right | \right |_1.
\end{align*}
Expanding this term gives the desired bound through an application of Cauchy-Schwarz:
\begin{align*}
    \frac{t}{s} \left| \left |\left (\sum\limits_{i=1}^m \frac{\alpha_i}{\lambda_i}D M^\top v_i \right ) \right | \right |_1 &\leq \frac{t}{Ls}\sum\limits_{x \in S} \mu|_S(x) \sum\limits_{i=1}^m \left|\alpha_i \langle x, v_i \rangle \right|\\
    &\leq \frac{t}{Ls}\sum\limits_{x \in S} \mu|_S(x)  \sqrt{\sum\limits_{i=1}^m \alpha_i^2}\sqrt{\sum\limits_{i=1}^m \langle x,v_i \rangle^2}\\
    &\leq \frac{t \|v\| }{Ls}\sum\limits_{x \in S} \mu|_S(x)  \|x\|\\
    & = O\left ( \norm{v} d\frac{t}{s}\right ),
\end{align*}
since $L=\Omega(1/d)$ and $\|x\|=1$ for all $x \in S$.

It remains to show that $d-O(k)$ eigenvectors have eigenvalue at least $\Omega(1/d)$. To see this, recall that since $(X,\mu)$ is in $1/4$-isotropic position, the eigenvalues of its covariance matrix are all at least $\frac{3}{4d}$. Because $S$ contains $\geq 1-O(k/d)$ fraction of $(X,\mu)$, we expect that $\text{Cov}(S)$ will also contain many large eigenvalues. Formally, we can bound the eigenvalues of $\text{Cov}(S)$ by writing it in terms of $\text{Cov}(X)$ and $\text{Cov}(X \setminus S)$. Let $A \succcurlyeq B$ denote that $A-B$ is positive semidefinite, then:
\begin{align*}
    \E{x\sim \mu|_{S}}{xx^\top} \succcurlyeq \mu(S) \cdot \E{x\sim \mu|_{S}}{xx^\top} &= \E{x\sim\mu}{xx^\top} - \mu(X-S) \cdot \E{x\sim\mu|_{X - S}}{xx^\top}\\
    & \succcurlyeq \frac{3}{4d} I - O\left(\frac{k}{d}\right) \E{x\sim\mu|_{X - S}}{xx^\top}. \label{eq:order}
\end{align*} 
Notice that the trace, and therefore the sum of the eigenvalues, of the right-hand term is $O(k/d)$. This means that the term can have at most $O(k)$ eigenvalues of size at least $1/2d$, which in turn implies that $\text{Cov}(S)$ must have at least $d-O(k)$ eigenvalues of size at least $\frac{1}{4d}$.
\end{proof}
Now that we have proven such a subspace always exists for $(k,t,s)$-structured pairs, we shift our attention to showing it can be found efficiently.

\SetKwFor{RepTimes}{repeat}{times}{end}
\begin{algorithm}[h!]
\SetAlgoLined
\KwIn{Pair $(X \subset \R^d, \mu)$ in $1/4$-isotropic position, a point $x_{\text{ref}} \in \mathbb{R}^d$, Margin oracle $\bigo_\lambda$, probability parameter $p$, gap parameters $\ell$ and $k$.}
\nonl \textbf{Output:} Subspace $V \subset \R^d$.\\

\RepTimes{$O(\log(1/p))$}{
    \textbf{Initialize:} $S' = \{\}$\\
    \RepTimes{$O(k \log d)$}{
        Randomly sample $S$ from $\mu$ of size $O(d/k)$.\\
        Use margin oracle to get a representative point $x_S=\bigo_\lambda(S)$.\\
        \If{$|m(x_S,h)| \leq \frac{\lambda^2}{\ell} |m(x_{\text{ref}},h)|$}{
            $S' = S' \cup S$
        }
    }
    \If{$\Ebb_{x\in S'}[xx^\top]$ has $d - O(k)$ eigenvalues greater than $\Omega(1/d)$}{
        \Return{$V$, the span of eigenvectors with eigenvalues greater than $\Omega(1/d)$.} 
    }
}
\Return{\q{Failure}}
 \caption{\textsc{DimReduce}}
 \label{alg:dimreduce}
\end{algorithm}

\begin{lemma}[\textsc{DimReduce}]\label{lemma:subspace-algorithm}
Given a pair $(X \subset \R^d,\mu)$, a point $x_{\text{ref}} \in \mathbb{R}^d$, and parameters $p>0,\ell>2$, and $k$ such that $k<d/10$, there exists an algorithm \textsc{DimReduce} with access to a margin oracle $\bigo_{\lambda}(\cdot)$ with the following guarantees. For any (unknown) hyperplane $h \in \mathbb{R}^d \setminus \{0\}$:
\begin{enumerate}
    \item \textsc{DimReduce} returns either a subspace $V$ satisfying
\begin{enumerate}
        \item $dim(V) = d-O(k)$
        \item $\forall v \in V: |m(v,h)| \leq O\left(\frac{\norm{v}d\lambda^3}{\ell}|m(x_{\text{ref}},h)|\right)$
\end{enumerate}
or reports failure.
    \item If $(X,\mu)$ is $\left (k,t,s \right)$-structured for $t=O \left (\frac{|m(x_{\text{ref}},h)|}{\lambda\sqrt{d}} \right )$ and $s=O\left(\frac{\ell}{\lambda^2\sqrt{d}}\right)$, \textsc{DimReduce} returns a  subspace with probability at least $1-p$.
\end{enumerate}
Further, \textsc{DimReduce} satisfies the following query guarantees:

\begin{enumerate}
    \item \textsc{DimReduce} makes at most $O(k\log(d)\log(1/p))$ queries
    \item \textsc{DimReduce} makes at most $O(k\log(d)\log(1/p))$ calls to $\bigo_{\lambda}(\cdot)$.
\end{enumerate}
\end{lemma}
\begin{proof}
In \Cref{lemma:subspace-structure}, we showed how to build $V$ using the covariance matrix of $S$, the set of all small margin points in $X$. In reality, we do not have access to $S$. Instead, we show that it is possible to sample from $S^*$, a set analogous to $S$, and that the covariance matrix of this sample will be sufficient for our purposes with high probability. First, we define $S^*$ and a relaxation $\widetilde S^* $:
\begin{align*}
    S^* &= \left \{ x \in X: |m(x,h)| \leq \frac{t}{s}\right \},\\
    \widetilde S^* &= \left \{ x \in X: |m(x,h)| \leq \frac{\lambda^2t}{s}\right \}.
\end{align*}
Next, we show that it is possible to efficiently find a sample of $O(d\log(d))$ points from $\widetilde S^*$ that contains a random sample from $S^*$ with constant probability. To start, notice that if $(X,\mu)$ is $\left (k,t,s \right)$-structured, then a random sample $S$ from $\mu$ of size $O(d/k)$ has noticeable probability of lying entirely in $S^*$:
\begin{equation}
\label{eq:constant-prob}
    \Pr\left[\forall x \in S: |m(x,h)| \leq \frac{t}{s}
    \right] \geq e^{-3}.
\end{equation}
Further, given a sample $S$, we can use the margin oracle to verify whether $S \subset \widetilde S^*$. Let $\bigo_{\lambda}(S)=x_S$, then using three linear queries, we can check the inequality:
\[
|m(x_S,h)| \leq \frac{\lambda^2}{\ell} |m(x_{\text{ref}},h)|.
\]
If this holds, then we have verified that:
\begin{align}\label{eq:small-ver}
\sqrt{\sum\limits_{x \in S} m(x,h)^2} \leq \lambda |m(x_S,h)| \leq \frac{\lambda^3}{\ell}|m(x_{\text{ref}},h)|
\end{align}
and thus that every point in $S$ lies in $\widetilde S^\star$.
Further, if points in $S$ have small enough normalized margin, the query will always be successful. In particular, as long as
\begin{equation}\label{eq:small-comp}
    \sqrt{\sum_{x\in S} m(x,h)^2} \leq \frac{\lambda}{\ell} |m(x_{\text{ref}},h)|,
\end{equation} 
our queries will verify \Cref{eq:small-ver}. Equation \eqref{eq:constant-prob} proves not only that this occurs with constant probability, but also that $S$ will be a random sample from $S^*$ with probability at least $e^{-3}$. Consider building a sample $S'$ by repeating this process $O(k\log(d))$ times and taking the union of samples such that \Cref{eq:small-comp} holds, forcing $S' \subset \widetilde S^*$. By Chernoff bounds, with constant probability $|S'|=O(d\log(d))$, and a constant fraction of the sub-samples making up $S'$ are random samples from $\mu|_{S^*}$.

To show that the uniform covariance matrix of $S'$, $\mathbb{E}_{x \in S'}[xx^\top]$, has the desired properties, first consider a random sample $S \subset S^*$ drawn from $\mu|_{S^*}$. As long as $S$ is sufficiently large, we claim the uniform covariance matrix of $S$, $\mathbb{E}_{x \in S}[xx^\top]$, is close to $\text{Cov}(S^*)$. To see this, first note that the spectral norm $\norm{\cdot}_s$ of this matrix is sandwiched by:
\begin{align*}
    \frac{1}{4d} \leq \norm{\mathbb{E}_{x \sim \mu|_{S^*}}[xx^\top]}_s \leq \frac{5}{2d}.
\end{align*}
Here the lower bound follows from the proof of \Cref{lemma:subspace-structure}, and the upper bound follows from noticing that for any unit vector $v$:
\begin{align*}
    v \E{x\sim \mu|_{S^\star}}{xx^\top} v^\top &= \frac{1}{\mu(S^\ast)}\sum_{i = 1}^{|S^\ast|}\mu(s_i)\inner{v}{s_i}^2\\
    &\leq  \frac{5}{4d \mu(S^\ast)} \leq \frac{5}{4(d - 5k)} < \frac{5}{2d}.
\end{align*}
Then for $S$ sufficiently large, the matrix Bernstein inequality~\cite[Example 1.6.3]{tropp2015introduction}  implies:
\begin{align*}
 \Pr_S\left[\norm{\E{x\in S}{xx^\top} - \E{x\sim \mu|_{S^\star}}{xx^\top}}_s <  \frac{1}{8d}\right]\geq \frac{1}{2}.
\end{align*}
If this holds, then the uniform covariance matrix of $S$ has at least $d-O(k)$ eigenvalues of size $\Omega(1/d)$. 

Recall that with constant probability, a constant fraction of our sample $S' \subset \widetilde S^*$ is a random sample $S \subset S^*$ drawn from $\mu$. By the above, the covariance matrix of this subsample $S$ has with constant probability at least $d-O(k)$ eigenvalues of size $\Omega(1/d)$. We show that the additional samples in $S' \setminus S$ do not affect this too much. In particular, since $S \subseteq S'$ we have:
\begin{align*}
    \sum\limits_{x \in S'}xx^\top \succcurlyeq \sum\limits_{x \in S}xx^\top,
\end{align*}
and therefore 
\begin{equation*}
    \Ebb_{x\in S'}[xx^\top] \succcurlyeq \frac{|S|}{|S'|}\Ebb_{x\in S}[xx^\top].
\end{equation*}
Since $S$ is a constant fraction of $S'$, we get that $S'$ has $d-O(k)$ eigenvalues of size at least $\Omega(1/d)$. Noting that we can simply check this condition manually, repeating this process $O(\log(1/p))$ times ensures we find such a sample $S'$ with probability at least $1-p$. In this case, the algorithm outputs $V$ to be the span of the eigenvectors with eigenvalue $\Omega(1/d)$ of the uniform covariance matrix of $S'$. If after this many repetitions no such sample is found, the algorithm aborts and returns nothing.

Finally, note that since $S' \subset \widetilde S^*$, by the same argument as \Cref{lemma:subspace-structure} (setting $\mu$ to be uniform), we have that for any vector $v \in V$: 
\begin{align*}
    |m(v,h)| \leq O \left (\norm{v}d\frac{\lambda^2t}{s} \right )
\end{align*}
as desired. Noting that finding each candidate sample takes only $O(k\log(d))$ queries and oracle calls and that we repeat this process at most $O(\log(1/p))$ times gives the desired query and oracle complexity.
\end{proof}
\subsubsection{Inference}
\begin{algorithm}[h!]
\SetAlgoLined
\KwIn{Pair $(X \subset \R^d, \mu)$ in $1/4$ isotropic position, Margin oracle $\bigo_\lambda$, and probability parameter $p$.}
\nonl \textbf{Output:} Partial labeling of $X$.\\
\nonl \textbf{Constants:} Gap parameter $\ell \geq \Omega(d^{5/2}\lambda^4)$, probability parameter $p'=p/2$.
\\
Sample $|X|$ coefficients $g_i \sim \mathcal N(0,1)$.\\
\If{$\sign \left( \left\langle \sum\limits_{i=1}^{|X|} x_ig_i, h \right \rangle \right ) = 0$}{
    \If{$\sign(\langle e_i, h \rangle)= 0$ for all standard basis vectors $e_i$}{
    Label all points in $X$ as $0$.\\
    \Return{}
    }
}
Run \textsc{StructureSearch} (\Cref{alg:find-gap}) with parameters $\ell$ and $p'$ to find a parameter $k$ and reference point $x_{\text{ref}}$.\\
Initialize $B$ as standard basis of $\R^d$. \\
\If{$k < d/10$}{
    Run \textsc{DimReduce} (\Cref{alg:dimreduce}) to find $V$, Abort if \textsc{DimReduce} returns nothing.\\
    Set $B$ to be an orthonormal basis of $V^\perp$.
}
For each $w_i \in B$, estimate relative margin $\gamma_i$ wrt the reference point $x_{\text{ref}}$ up to error $\frac{1}{3\sqrt{10}\lambda  d^{3/2}}$.\\
\For{$x\in X$}{
    For each $w_i \in B$, compute coefficients $\beta_i = \langle x,w_i \rangle$.\\
    \If{$\left|\sum\limits_{i=1}^k \beta_i\gamma_i\right| \geq \frac{2}{3\lambda \sqrt{10d}}$}{
        Label $x$ as $\sign(\langle x_{\text{ref}}, h \rangle ) \cdot \sign\left(\sum\limits_{i=1}^{|B|} \beta_i\gamma_i \right)$.
    }
}

 \caption{\textsc{IsoLearn}}
 \label{alg:isolearn}
\end{algorithm}

With these structural and algorithmic lemmas out of the way, we can finally present \textsc{IsoLearn} (\Cref{alg:isolearn}), a method that infers all large margin points with high probability:
\begin{proof}[Proof of \Cref{lemma:isotropic-oracle-learner}]
Recall that our previous lemmas work only in the case that $h$ is non-zero. To test for this degerate case, we check the sign of a random Gaussian combination of $X=\{x_1,\ldots,x_n\}$:
\[
x_X = \sum\limits_{i=1}^{n} x_ig_i,
\]
where $g_i$ are independent draws from a standard normal distribution. In particular, notice that if $h=0$, $\sign(x_X)=0$. On the other hand, if $h \neq 0$, then with probability $1$, $\sign(x_X) \neq 0$. If $\sign(x_X)=0$, we verify that $h=0$ by checking that each standard basis vector $e_i$ satifies $\sign(\langle e_i, h \rangle)=0$, and return that all points in $X$ have label $0$ if this succeeds. Otherwise, it must be the case that $h$ is non-zero, putting us in position to run \textsc{StructureSearch} (\Cref{lemma:t-k}).

With parameters $\ell \geq \Omega(d^{5/2}\lambda^4)$ and $p' = p/2$, run \textsc{StructureSearch} to find a parameter $k$ and reference point $x_{\text{ref}}$ such that with probability at least $1-p/2$, $(X,\mu)$ is $\left(k,t,s\right)$-structured for $t = \frac{\langle  x_{\text{ref}}, h \rangle}{\norm{h}\sqrt{10d}\lambda}$ and $s = \frac{\ell}{\lambda^2\sqrt{10d}}$. Assume for the time being that $k < d/10$, we note at the end how to adapt to the case that $k \geq d/10$. Then with probability at least $1-p/2$, the following three properties hold:
\begin{align}
    \label{eq:bound-size}
    |m(x_{\text{ref}},h)|&\geq (d\lambda)^{-O(\log d)}\\
    \Pr_{x \sim \mu}\left [|m(x,h)| \geq t \right] &\geq \frac{k}{d}\\
    \Pr_{x \sim \mu}\left [|m(x,h)| \leq t/s\right] &\geq 1-\frac{5k}{d}.
\end{align}
Let $e=10k$. If these assumptions hold, \textsc{DimReduce} (\Cref{lemma:subspace-algorithm}) returns with probability at least $1-p/2$ a subspace $V$ of dimension $d-e$ such that for any unit vector $v \in V$:
\begin{align}\label{eq:inf-small}
    |m(v,h)| \leq \frac{t}{12d}.
\end{align}
 For our inference, we will require a slightly weaker claim (by a factor of $4$) than \Cref{eq:inf-small}:
\begin{align}\label{eq:inf-small-2}
    |m(v,h)| \leq \frac{t}{3d}.
\end{align}
We use this second claim for our inference due to our later strategy for creating a zero-error LDT (\Cref{thm:zldt}). This strategy requires there to be a slight gap between the bound that is needed for inference (\Cref{eq:inf-small-2}), and the bound that is true with high probability (\Cref{eq:inf-small}).

Turning our attention back to the task at hand, the idea behind our inference is that vectors in $V$ do not have a large effect on the sign of large margin points. With this in mind, pick an orthornormal basis $(v_1,\ldots,v_{d-e})$ for $V$, and extend it to an orthornormal basis of $\mathbb{R}^d$ via $(w_1,\ldots,w_{e}) \in V^{\perp}$. We can express any point $y \in X$ as:
\begin{align}\label{eq:y-infer-1}
y = \sum\limits_{i=1}^{d-e} \alpha_i v_i + \sum\limits_{i=1}^{e} \beta_i w_i
\end{align}
where $|\alpha_i|,|\beta_i| \leq 1$. In particular this means that we can understand the margin of $y$ in terms of the margin of our basis vectors:
\begin{align}\label{eq:y-infer-2}
m(y,h) = \sum\limits_{i=1}^{d-e} \alpha_i m(v_i,h) + \sum\limits_{i=1}^{e} \beta_i m(w_i,h).
\end{align}
Notice that the left-hand sum is bounded via \Cref{eq:inf-small-2} with respect to our reference point $x_{\text{ref}}$. Further, since the right-hand sum only has $O(k)$ terms, we can afford to estimate each term's relative margin with $x_{\text{ref}}$ up to very high accuracy. In particular, based on the assumption that $m(x_{\text{ref}},h) \geq (d\lambda)^{-O(\log(d))}$, we can estimate the relative margin up to:
\[
m(w_i,h) = m(x_{\text{ref}},h)\left (\gamma_i \pm \frac{1}{3\sqrt{10}\lambda  d^{3/2}} \right ).
\]
using at most $\log(|m(x_{\text{ref}},h)|) \leq O(\log(\lambda d)\log(d))$ linear queries for each of the $O(k)$ basis vectors.
Combining these bounds with equations \eqref{eq:inf-small-2} and \eqref{eq:y-infer-2} upper and lower bounds the normalized margin of $y$:
\begin{align}\label{eq:y-infer-3}
m(y,h) \in m(x_{\text{ref}},h) \left (\sum\limits_{i=1}^{k} \beta_i\gamma_i \pm  \frac{1}{3\lambda\sqrt{10d}}\right ).
\end{align}
Thus to know the sign of $y$, it is sufficient to know that $\left|\sum\limits_{i=1}^k \beta_i\gamma_i\right| > \frac{1}{3\lambda \sqrt{10d}}$. Notice, in fact, that if we have a stronger guarantee, $\left|\sum\limits_{i=1}^k \beta_i\gamma_i\right| \geq \frac{2}{3\lambda \sqrt{10d}}$, we additionally learn a 2-approximation of the relative margin of $y$ to $x_{\text{ref}}$, that is:
\begin{align}\label{eq:approx-relative}
\frac{1}{2}\left |\sum\limits_{i=1}^k \beta_i\gamma_i \right | \leq \left|\frac{\langle y, h \rangle}{\langle x_{\text{ref}}, h \rangle}\right| \leq 2\left |\sum\limits_{i=1}^k \beta_i\gamma_i \right |.
\end{align}
While not important at the moment, this approximation is a key factor for our zero-error LDT that follows, and is also the reason we need the gap between Equations~\eqref{eq:inf-small} and \eqref{eq:inf-small-2} discussed above.

Finally, we claim that this process infers every point with normalized margin at least $t$, which, according to \Cref{lemma:t-k}, is a $\frac{k}{d}$ fraction of $(X,\mu)$ with high probability. Assume that $y$ satisfies $
|m(y,h)| \geq t$. Then we have:
\begin{align*}
   t &\leq |m(y,t)|\\
   & \leq |m(x_{\text{ref}},h)|\left (\left |\sum\limits_{i=1}^{k} \beta_i\gamma_i\right | + \frac{1}{3\lambda \sqrt{10d}} \right )\\
   & = t\lambda \sqrt{10d} \left (\left |\sum\limits_{i=1}^{k} \beta_i\gamma_i\right | +  \frac{1}{3\lambda \sqrt{10d}} \right ),
\end{align*}
which implies that:
\begin{align*}
    \left |\sum\limits_{i=1}^{k} \beta_i\gamma_i\right | & \geq \frac{2}{3\lambda \sqrt{10d}}.
\end{align*}
Now we briefly turn our attention to case in which $10k \geq d$. Here the proof is largely the same, except we take $V^{\perp}$ to be all of $\mathbb{R}^d$. In this case, we infer a $k/d$ fraction of points for $k=\frac{d}{10}$. 

Finally, together \textsc{StructureSearch} and \textsc{DimReduce} succeed with probability $\geq 1-p$, and use at most $O(k\log(d)\log( d/p))$ queries and oracle calls. To learn the basis of $V^\perp$, our learner makes an additional $O(k\log(d)\log(d\lambda))$ queries, bringing the total to at most $O(k\log(d)\log(\lambda d/p))$. If at any point we would make more than $O(d\log(d)\log(\lambda d/p))$ (this may occur if \textsc{StructureSearch} fails), we abort and declare the learner has failed, outputting nothing.
\end{proof}
\subsection{A Weak Learner for Arbitrary Distributions}\label{sec:arb-ww-learner}
We now show how to transform \textsc{IsoLearn}, our weak learner for $1/4$-isotropic pairs, into \textsc{PartialLearn}, our weak learner for arbitrary pairs. The proof of this Lemma relies heavily on the work of Barthe~\cite{barthe1998reverse}, restated by Dvir, Saraf, and Wigderson~\cite{dvir2014breaking}. In particular, Barthe, Dvir, Saraf, and Wigderson provide a sufficient condition for $(X,\mu)$ to be transformed into $\varepsilon$-isotropic position. To understand their result, we introduce some useful terminology.
\begin{definition}[Transformed Pair]
Given a pair $(X \subset \R^d,\mu)$ and an invertible linear transformation $T: \mathbb{R}^d \to \mathbb{R}^d$, we define the transformed pair $(X_T,\mu_T)$ to be:
\[
X_T = \left \{ x_T = \frac{T(x)}{\norm{T(x)}}: x \in X \right\},~\mu_T(x_T) = \mu(x).
\]
\end{definition}
The condition of Barthe, Dvir, Saraf, and Wigderson depends on the set of bases of $X$, $\mathcal{B}(X)$, and its convex hull $K(X)$ which we define next.
\begin{definition}
Let $X \subset \mathbb{R}^d$ be a set. We denote by $\mathcal{B}(X)$ the set of subsets $B \subset X$ which form a basis of $\R^d$. Considering each basis $B$ as an indicator function $1_B \in \{0,1\}^X$, we let $K(X) \subset \R^{X}$ denote the convex hull of $\mathcal{B}(X)$.
\end{definition}
With these definitions, we can state the sufficient condition of \cite{barthe1998reverse,dvir2014breaking} for transforming a pair $(X,\mu)$ into $\varepsilon$-isotropic position.
\begin{lemma}[Barthe's Theorem: Lemma 5.1~\cite{dvir2014breaking}]\label{lemma:Barthe} Given a pair $(X \subset \R^d,\mu)$, if the vector $d\mu$ is in $K(X)$, then for all $\varepsilon > 0$, there exists an invertible linear transformation $T:\R^d \to \R^d$ such that the corresponding transformed pair $(X_T,\mu_T)$ is in $\varepsilon$-isotropic position.
\end{lemma}
We note that Lemma $5.1$ in~\cite{dvir2014breaking} was only for the uniform distribution. However, with minor modifications the same proof gives \Cref{lemma:Barthe} for general distributions.
Since point location is invariant to invertible linear transformations of the data (see \Cref{sec:overview-arb-ww-learner}), we are in good shape to apply \textsc{IsoLearn} as long as $(X,\mu)$ satisfies this condition. 

However, we are interested in arbitrary sets endowed with an arbitrary distribution, for which this need not be the case. To circumvent this, we prove a novel structural result: for every pair $(X,\mu)$, there exists a subspace $V$ dense in $X$ such that $X \cap V$ may be transformed into approximate isotropic position. Most of the work in proving this lies in showing a new characterization of when a pair $(X,\mu)$ may be transformed into $\varepsilon$-isotropic position. In particular, we show that this is possible if and only if every $k<d$ dimensional subspace contains at most a $\frac{k}{d}$ fraction of $(X,\mu)$.
\begin{lemma}\label{prop:dense-subspace}
Given a pair $(X \subset \R^d,\mu)$, the following conditions are equivalent:
\begin{enumerate}
    \item For any $\varepsilon > 0$ there exists an invertible linear map $T$ such that the pair $(X_T,\mu_T)$ is in $\varepsilon$-isotropic position.
    \item For every $1 \leq k \leq d$, every  $k$-dimensional subspace $V$ satisfies $\mu(V \cap X) \leq \frac{k}{d}$.
\end{enumerate}
\end{lemma}

\begin{proof}
We begin by proving the contrapositive of the forward direction. In particular, assume there exists some $k<d$ dimensional subspace $V$ such that $\mu(V \cap X) = \frac{k}{d}+\delta$ for some $\delta>0$. Then  for any invertible linear transformation $T$, $\mu_T(T(V) \cap X_T) = \frac{k}{d} + \delta$. The covariance matrix of $X_T$ restricted to the subspace $T(V)$ then has trace at least $\frac{k}{d} + \delta$, which implies that some eigenvalue must be greater than $\frac{1}{d} + \frac{\delta}{k}$. In particular, it cannot be in $\varepsilon$-isotropic position for $\varepsilon<\delta d /k$.

The backward direction is more involved. Again we prove the contrapositive. Assume that for some $\varepsilon>0$, $(X,\mu)$ cannot be transformed into $\varepsilon$-isotropic position.
Thus, by \Cref{lemma:Barthe}, $d\mu \notin K(X)$, and moreover there exists a hyperplane separating the two. Since each basis indicator $1_B$ is in $K(X)$, this implies the existence of some normal vector $w \in \R^X$ such that for all bases $B \in \mathcal{B}(X)$:
\begin{align}\label{eq:basis-weight}
    \sum\limits_{x \in B} w_x < d \sum\limits_{x \in X}\mu(x)w_x.
\end{align}
Assume for the sake of contradiction that for all $1 \leq k < d$, no $k$-dimensional subspace contains more than a $\frac{k}{d}$ probability mass of $X$ with respect to $\mu$. Using this assumption, we will build a basis that violates \Cref{eq:basis-weight}. Let $X=\{x_1,\ldots,x_n\}$ denote a sorted order in which for all $i$, $w_{x_i} \ge w_{x_{i+1}}$. We choose our basis greedily from this order. In particular, say that we have already chosen points $x_{j_1},\ldots,x_{j_{i-1}}$, and would like to select the $i$th point for our basis. Our strategy is simply to pick from the available points the one with the largest possible weight:
\[
j_i = \min\{ j : x_j \not\in \text{Span}( x_{j_1},\ldots,x_{j_{i-1}})  \}.
\]
Our goal is to prove that for this construction:
\[
 \sum\limits_{i=1}^d w_{x_{j_i}} \geq d \sum\limits_{x \in X}\mu(x)w_x.
\]
To see why this might be the case, consider the indices $\mu_i, 1 < i \leq  d$, which partition $X$ into segments with approximately equal measure:
\[
\mu_i = \min \left \{ j \in [n] : \sum\limits_{k=1}^j w_{x_k} > \frac{i-1}{d} \right \}.
\]
For notational convenience, let $\mu_1=1$, and $\mu_{d+1}=n+1$.
The idea is that the points in our greedy basis must have at least as much weight as the $w_{x_{\mu_i}}$, since by assumption, the span of $\{x_{j_1},\ldots,x_{j_{i-1}}\}$ can have at most measure $\frac{i-1}{d}$. In other words, it must be the case that:
\begin{align}\label{eq:basis-weight-2}
    \sum\limits_{i=1}^d w_{x_{j_i}} \geq \sum\limits_{i=1}^d w_{x_{\mu_i}}.
\end{align}
Informally, each $w_{\mu_i}$ corresponds to the largest weight in disjoint $1/d$ measure segments of $(X,\mu)$, so the sum of the weights must be at least $d$ times the average, providing our contradiction. In reality, however, the proof is complicated slightly by the fact that each segment may not have measure exactly $1/d$. In more detail, consider the partition defined by the $\mu_i$'s, $X=X_1 \amalg \ldots \amalg X_d$, where each $X_i$ is:
\[
X_i = \{x_{\mu_{i}},\ldots,{x_{\mu_{i+1}-1}} \}.
\]
Assuming no single point has measure greater than $1/d$ (which by itself would provide a contradiction), the $X_i$ are non-empty. To make the measure of each segment exactly $1/d$ so we may apply the reasoning above, we slightly modify each set by splitting up the measure of $x_{\mu_i}$ between $X_i$ and $X_{i-1}$. In particular, for each $1 < i \leq d$, we define two copies $x_{\mu_i}^1$, and $x_{\mu_i}^2$, where the mass of $x_{\mu_i}^1$ is:
\[
\mu(x_{\mu_i}^1) = \frac{i}{d} - \mu(X_1 \cup \ldots \cup X_{i-1}),
\]
and $x_{\mu_i}^2$ has the remaining mass:
\[
\mu(x_{\mu_i}^2) = \mu(x_{\mu_i}) - \mu(x_{\mu_i}^1).
\]
Let $X'$ be the set resulting from replacing each $x_{\mu_i}$ with $\{x_{\mu_i}^1,x_{\mu_i}^2\}$. We define a new partition $X'= X_1' \amalg \ldots \amalg X_d'$:
\[
X_i' = \{x_{\mu_i}^1\} \cup \left( X_i \setminus \{x_{\mu_i}\}  \right) \cup \{x_{\mu_{i+1}}^2\},
\]
where each set now has measure exactly $\frac{1}{d}$. Associating the weight $w_{x_{\mu_i}}$ to both $x^1_{\mu_i}$ and $x^2_{\mu_i}$, notice that this modification preserves both order by weight and average weight, ensuring the following equations hold:
\begin{align*}
    w_{x_{\mu_i}} &\geq d \sum\limits_{x' \in X'_i} \mu(x')w_{x'},\\
    \sum\limits_{x \in X}\mu(x)w_x &= \sum\limits_{x' \in X'}\mu(x')w_{x'}.
\end{align*}
Together with \Cref{eq:basis-weight-2}, these allow us to derive our contradiction:
\begin{align*}
    \sum\limits_{i=1}^d w_{x_{j_i}} &\geq \sum\limits_{i=1}^d w_{x_{\mu_i}}\\
    & \geq \sum\limits_{i=1}^d d\sum\limits_{x' \in X_i'} \mu(x')w_{x'}\\
    & = d\sum\limits_{x' \in X'} \mu(x')w_{x'}\\
    & = d\sum\limits_{x \in X} \mu(x)w_x
\end{align*}
\end{proof}
Noting that the identity acts as the desired transform for any $(X,\mu)$ in one dimension, \Cref{overview:prop-iso} follows as a corollary by induction:
\begin{corollary}[\Cref{overview:prop-iso}]\label{cor:iso}
Given a pair $(X \subset \R^d,\mu)$, for some $1 \leq k \leq d$, for all $\varepsilon > 0$ there exist:
\begin{enumerate}
    \item A $k$-dimensional subspace $V$ with the property $\mu(X \cap V) \geq \frac{k}{d}$.
    \item An invertible linear transformation $T:V \to V$ such that the pair $((X \cap V)_T, (\mu|_{X \cap V})_T)$ is in $\varepsilon$-isotropic position,
\end{enumerate}
where $\mu|_{X \cap V}$ denotes the normalized restriction of $\mu$ to $X \cap V$.
\end{corollary}
\begin{proof}
We induct on the dimension $d$. Our base case, $d=1$, follows trivially from setting $T$ to be the identity. For the inductive step assume that the result holds up to dimension $d-1$. We split into two cases based on the existence of an invertible transformation $T$ such that $(X_T, \mu_T)$ is in $\varepsilon$-isotropic position:
\paragraph{Case 1: $T$ exists.} Setting $k=d$, we are done by assumption.
\paragraph{Case 2: $T$ does not exist.}  Applying \Cref{prop:dense-subspace}, for some $1 \leq m < d$ there exists an $m$-dimensional subspace $V_m$ such that $\mu(V_m \cap X) > \frac{m}{d}$. Since the set $V_m \cap X$ lies in $m < d$ dimensions, we can apply our inductive hypothesis, that there exists a $k$-dimensional subspace $V_k \subseteq V_m$ such that:
\begin{enumerate}
    \item $\mu(V_k \cap (V_m \cap X)) = \mu(V_k \cap X) \geq \frac{k}{m}\mu(V_m \cap X) > \frac{k}{d}$ and,
    \item $V_k \cap X$ can be transformed into $\varepsilon$-isotropic position,
\end{enumerate}
which completes the proof.
\end{proof}
In \Cref{app:iso}, we provide a slightly more complicated characterization for the existence of exact isotropic transforms. While the exact result is unnecessary for our work, it generalizes the work of Forster \cite{forster2002linear} and may be of independent interest. 

Finally, we are in position to prove \Cref{thm:arbitrary-ww-learner}.
\begin{proof}[Proof of \Cref{thm:arbitrary-ww-learner}]
By \Cref{cor:iso}, there exists a subspace $V$ of dimension $1 \leq k \leq d$ with the properties:
\begin{enumerate}
    \item $\mu(V \cap X) \geq \frac{k}{d}$
    \item There exists an invertible linear transformation $T: V \to V$ such that $((X\cap V)_T,(\mu|_{X \cap V})_T)$ is in $1/4$-isotropic position.
\end{enumerate}
Kane, Lovett, and Moran~\cite{kane2018generalized} observe that point location is invariant to invertible linear transformations. In more detail, let $T$ be such a transformation, $h' = (T^{-1})^{\top}(h)$, and $x' = \frac{T(x)}{\norm{T(x)}}$. Observe that $\langle x',h' \rangle = \langle \frac{x}{\norm{T(x)}}, h \rangle$. Thus not only is it sufficient to learn the labels of $x'$ with respect to $h'$, but we can simulate linear queries on $h'$ simply by normalizing $x$ by an appropriate constant.

In other words, we are free to act as if we are learning over $((X\cap V)_T,(\mu|_{X \cap V})_T)$, which is in $1/4$-isotropic position. Restricting to the $k$-dimensional subspace $V$, run \textsc{IsoLearn}\footnote{For simplicity we presented \text{IsoLearn} over $\mathbb{R}^k$, but it easy to see it may be performed over any $k$-dimensional Euclidean space.} with probability parameter $p/2$ using a simulated margin oracle $\bigo_{\lambda}$ for  $\lambda=\frac{\text{poly}(d)}{p}$. By \Cref{lemma:simulate-margin-oracle}, union bounding over the at most $O(d\log(d)\log(d/p))$ calls to $\bigo_{\lambda}$ shows that with probability at least $1-p/2$, all oracle calls will be successful. Since \textsc{IsoLearn} is otherwise $(1-p/2)$-reliable, the resulting learner is $p$-reliable. Further, with probability $1-p/2$ there is some parameter $1 \leq m \leq k$ such that \textsc{IsoLearn} has coverage at least $\frac{m}{k}$ on $(X \cap V, \mu|_{X \cap V})$ and makes at most $O(m\log(k)\log(k\lambda/p))$ queries. Since $\mu(X \cap V) \geq \frac{k}{d}$, this implies that the resulting learner's coverage over $(X,\mu)$ is at least $\frac{m}{d}$ as desired. The remaining properties follow directly from \textsc{IsoLearn}.
\end{proof}
\section{Bounded-error LDT: Boosting}\label{sec:boosting}
In this section, we provide the details of how to apply a boosting procedure to \textsc{PartialLearn} in order build a randomized LDT $T$ that $\delta$-reliably computes point location. Informally, recall that this means that for any hyperplane $h$, with probability at least $1-\delta$, $T$ labels \textit{every} point in $X$ correctly.
\begin{theorem}[\Cref{overview:bldt}]\label{thm:bldt}
Let $X \subset \mathbb{R}^d$, $|X|=n$. Then there exists a randomized LDT $T$ that $\delta$-reliably computes the point location problem on $X$ with maximum depth:
\[
\MD(T) \leq O\left(d\log^2(d)\log\left(\frac{n}{\delta}\right)\right).
\]
\end{theorem}
This theorem follows from the combination of two boosting procedures. To simplify the process, we first apply a standard boosting process from~\cite{kane2018generalized} to create a $.01$-reliable learner for arbitrary $(X,\mu)$ that has $99\%$ coverage with $99\%$ probability.
\begin{algorithm}[h!]
\SetAlgoLined
\KwIn{Pair $(X \subset \R^d, \mu)$.}
\nonl \textbf{Output:} Partial labeling of $X$.\\
\nonl\textbf{Constants:} $p = 1/\poly(d)$, constant $c$ such that \textsc{PartialLearn} uses no more than $ck\log^2(d)$ queries with probability $1-p$.\\
\textbf{Initialize:} $i=0$, $Q = 0$ and $X_0 = X$\\
\While{$Q \leq 5cd \log^2 (d)$ and $X_i \neq \{\}$}  {
    Run \textsc{PartialLearn} on $(X_i,\mu|_{X_i})$.\\
    Set $X_{i+1}$ to be the set of un-inferred points in $X_i$.\\
    $i \gets i+1$
}
 \caption{\textsc{WeakLearn}}
 \label{alg:weaklearner}
\end{algorithm}
\begin{lemma}[\textsc{WeakLearn} (\Cref{overview:w-leaner})]\label{lemma:weakLearner}
Given a pair $(X \subset \R^d,\mu)$, there exists a weak learner \textsc{WeakLearn} with the following guarantees. For any (unknown) hyperplane $h \in \R^d$:
\begin{enumerate}
    \item \textsc{WeakLearn} is $.01$-reliable.
    \item With probability at least $.99$, \textsc{WeakLearn} has coverage at least $.99$.
    \item \textsc{WeakLearn} makes at most $O(d\log^2(d))$ queries.
\end{enumerate}
\end{lemma}
\begin{proof}
We apply the boosting procedure of~\cite{kane2018generalized}--restricting each round to the set of un-inferred points.
In particular, setting the probability parameter $p=1/\text{poly(d)}$, \textsc{PartialLearn} provides a $p$-reliable learner which, with probability $1-p$, for some $1\leq k \leq d$ has coverage at least $k/d$ while making fewer than $ck\log^2(d)$ queries for some constant $c>0$.
Consider the boosting process laid out in \Cref{alg:weaklearner}, that is setting $i = 0$, $X_0 = X$:
\begin{enumerate}
    \item Run \textsc{PartialLearn} (\Cref{thm:arbitrary-ww-learner}) on $(X_i,\mu|_{X_i})$.
    \item Set $X_{i+1}$ to be the set of un-inferred points in $X_i$. Set $i \leftarrow i+1$ and repeat.
\end{enumerate}
If at any point in the process more than $5cd \log^2 (d)$ queries are used, we abort. We analyze the coverage and reliability of this process.
Let $k_i$ denote the $k$ parameter for the $i$th learner, and let $t$ denote the final complete iteration before the learner is aborted. In the event that each learner uses at most $ck\log^2(d)$ queries, we can lower bound the sum of the $k_i$:
\begin{align*}
    \sum\limits_{i=1}^t k_i \geq 5.
\end{align*}
Assume further that each weak learner has coverage $k_i/d$ and makes no errors. We can then bound the coverage of our boosted learner by:
\begin{align*}
    C_{\mu}(A) &\geq 1 - \prod\limits_{i=0}^{t}\left ( 1- \frac{k}{d}\right )\\
    & \geq 1- e^{\frac{1}{d}\sum\limits_{i=0}^{t} k_i}\\
    & \geq 1 - e^{-5}.
\end{align*}
Notice that since we run at most $5cd\log^2(d)$ weak learners, for sufficiently small $p$ union bounding over these events gives that with at least $99\%$ probability, the resulting boosted weak learner makes no errors and has at least $99\%$ coverage.
\end{proof}
Unfortunately, this boosting procedure is not efficient enough to use for learning all of $X$. We would be forced to set the correctness probability for our weak learner too low, costing additional factors. Instead, we employ a boosting procedure that relies on re-weighting learned points. 
\begin{lemma}[\textsc{Boosting}]\label{lemma:boosting}
Let $X \subset \mathbb{R}^d$ be a finite set of size $n$. If there exists for all distributions $\mu$ over $X$ a learner $A_{\mu}$ with the following guarantees:
\begin{enumerate}
    \item $A_{\mu}$ is $.01$-reliable.
    \item With probability at least $.99$, $A_{\mu}$ has coverage at least $.99$.
    \item $A_{\mu}$ uses at most $Q$ queries.
\end{enumerate}
Then there exists a randomized LDT $T$ that $\delta$-reliably computes the point location problem on $X$ with maximum depth:
\[
\MD(T) \leq O(Q\log(n/\delta))
\]
\end{lemma}
\begin{algorithm}[h!]
\SetAlgoLined
\KwIn{Set of points $X \subset \R^d$ and Margin oracle $\bigo_\lambda$.}
\nonl \textbf{Output:} Labeling of $X$.\\
\nonl\textbf{Constants:} number of iterations $T = O(\log(n/\delta))$.\\
\textbf{Initialize:} iteration $i = 0$, weights $\text{wt}(x) = 1$ and $\mu$ as the distribution induced by normalizing these weights.\\
Normalize points in $X$ to be unit vectors.\\ 
\While{$i \leq T$}  {
    Run \textsc{WeakLearn} (\Cref{alg:weaklearner}) on $(X,\mu)$ to receive partial labeling $A_i$.\\
    Update $\text{wt}(x) = \text{wt}(x)/11$ for all points $x \in X$ such that $A_i(x)\neq \bot$.\\
    Update $\mu$ according to the new weights.\\
    $i \gets i+1$
}
For each $x \in X$, label $x$ as the majority non-$\bot$ label from $\{A_i(x)\}_{i=1}^T$.
 \caption{\textsc{Boosting}}
 \label{alg:boosting}
\end{algorithm}
\begin{proof}
To begin, note that normalizing vectors in $X$ does not change their labels; we may therefore assume without loss of generality that $X$ consists of unit vectors. Assign each point $x \in X$ a weight $\text{wt}(x)$, initialized to $1$. Let the distribution over $X$ induced by normalizing these weights be denoted $\mu$.
We repeat the following strategy for $T=O(\log(n/\delta))$ iterations: 
\begin{enumerate}
    \item Run the weak learner $A_{\mu}$.
    \item Multiplicatively decrease the weight of any learned point by $\frac{1}{11}$.
\end{enumerate}
The idea behind exponentially decreasing the weights of learned points is that it forces each point to be labeled many times throughout this process. Since our weak learner is $.01$-reliable, most of these predictions must be correct, and in particular the majority prediction will be correct with high probability.

In greater detail, recall that our weak learner returns a partial classification with no mistakes with at least $99\%$ probability. Treating each run as an independent Bernoulli process, the probability that more than $2\%$ of runs have an error is at most $\delta/2$ by a Chernoff bound. If we can prove that each point will be labeled in at least $5\%$ of the iterations with probability at least $1-\delta/2$, then the majority label will be  correct for all points with probability $1-\delta$.

To prove that each point is labeled in at least $5\%$ of the iterations, notice that because our learner has $99\%$ coverage with at least $99\%$ probably, a Chernoff bound gives that at least $98\%$ of runs have $99\%$ coverage with probability at least $1-\delta/2$. Assuming this is the case, each run with $99\%$ coverage must reduce the total weight by at least $.01 \cdot 1 + 0.99 \cdot 1/11 = 1/10$. This means we can upper bound the total weight of $X$ after our process finishes by:
\[
\sum\limits_{x \in X} \text{wt}(x) \le n\left (\frac{1}{10} \right )^{0.98T}.
\]
On the other hand, assume some point $x$ was labeled in fewer than $5\%$ of runs. Then we can lower bound wt(x) at the end of the process by:
\[
wt(x) \geq \left(\frac{1}{11} \right )^{.05T}.
\]
For $T=O(\log(n/\delta))$ sufficiently large, this provides a contradiction.
\end{proof}
Together, \textsc{WeakLearn} and \Cref{lemma:boosting} immediately imply the desired bounded-error LDT.
\begin{proof}[Proof of \Cref{thm:bldt}]
 \textsc{WeakLearn} satisfies the conditions required by \Cref{lemma:boosting} with parameter $Q$ at most $O(d\log^2(d))$.
\end{proof}
Solving the point location problem in this bounded-error model is also sufficient to build a nearly optimal algorithm for actively learning homogeneous hyperplanes, denoted $\mathcal H_d^0$, with membership queries. The idea is simple. The sample complexity of passively PAC-learning this class is well known due to classic results on learning and VC dimension~\cite{vapnik1974theory, Blumer, hanneke2016optimal}. Once we have drawn a sample, building a learner reduces to solving a point location problem.
\begin{corollary}[\Cref{intro:active} (homogeneous case)]\label{thm:active}
There exists an active learner for $(\mathbb{R}^d,\mathcal H_d^0)$ using only
\[
n(\varepsilon,\delta) = O\left(\frac{d+\log(1/\delta)}{\varepsilon}\right)
\]
unlabeled samples, and 
\[
q(\varepsilon,\delta) = O\left (d\log^2(d)\log\left ( \frac{n(\varepsilon,\delta)}{\delta} \right ) \right )
\]
membership queries.
\end{corollary}
\begin{proof}
Recall that $n(\varepsilon,\delta)$ is the sample complexity of learning $(\mathbb{R}^d,\mathcal H_d^0)$, the number of samples needed to passively learn the distribution and hidden classifier up to $\varepsilon$ error with probability $1-\delta$. Although our learner errs with some probability, we can set this $\delta$ to $\delta/2$, and learn the sample with probability at least $1-\delta/2$ as well. Thus our goal is simply to solve a point location problem with error at most $1-\delta/2$, which \Cref{thm:bldt} proves can be done in the desired number of queries. While \Cref{thm:bldt} uses ternary queries, in the homogeneous case these can be easily simulated via two binary queries: $Q_{x}(h)$ and $Q_{-x}(h)$.
\end{proof}
The full non-homogeneous version of \Cref{intro:active} follows from the same general argument combined with our generalization of \Cref{thm:bldt} to non-homogeneous hyperplanes and binary queries given in \Cref{sec:non-homogeneous}.
\section{Zero-error LDT: Verification}\label{sec:verification}
Up to this point, we have allowed our learner to err with low probability. While this is sufficient for the standard models in learning theory, point location in the computational geometry literature is often studied from the standpoint of zero-error. We now show how to employ a verification process for our learner that removes all errors at the cost of an additional additive $d^{1+o(1)}$ factor in the depth of our LDT. 
\begin{theorem}[\Cref{overview:zldt}]\label{thm:zldt}
Let $X \subset \mathbb{R}^d$ be an arbitrary finite set. Then there exists a randomized LDT $T$ that  reliably computes the point location problem on $X$ with expected depth:
\[
\ED(T) \leq O\left(d\log^2(d)\log\left(n\right)\right) + d \cdot 2^{O\left(\sqrt{\log(d)\log\log(d)}\right)}.
\]
\end{theorem}
Notice that for large enough $n$ the lefthand term of the query complexity dominates, making our algorithm optimal in this regime up to factors logarithmic in the dimension. Recall that inferences made by \textsc{PartialLearn} only rely on a single unverified assumption, that the ``small margin'' points used to find our low margin subspace are indeed small as compared to our ``large margin'' reference point. We will prove that verifying these inequalities may be reduced to a combinatorial problem we call matrix verification. We recall the definition from \Cref{sec:overview-verification}.
\begin{definition}[Matrix Verification]
Let $S \subset \mathbb{R}^d$ be a subset of size $m$, $h \in \mathbb{R}^d$ a hyperplane, and $\{C_{ij}\}_{i,j=1}^m$ a constraint matrix in $\R^{m \times m}$. We call the problem of determining whether for all $i,j$:
\[
\langle  x_i, h \rangle \leq  C_{ij}\langle x_j, h \rangle
\]
a \textbf{matrix verification problem} of size $m$. Further, we denote by $V(m)$ the minimum expected number of queries made across randomized algorithms which solve verification problems of size $m$ in any dimension.
\end{definition}
We will actually use a slightly more general version of matrix verification in which entries in the constraint matrix may be empty, requiring no comparison. However, since this is a strictly easier problem, we will focus on the version presented above. Matrix verification and point location are closely related problems. Here, we show a two way equivalence: without much overhead, we can reduce point location to a small verification problem, and likewise may reduce matrix verification to several point location problems in fewer dimensions. Together, these observations set up a recurrence that allows us to solve both problems efficiently. Before proving these results, we introduce a useful notation for the worst case expected depth of point location.
\begin{definition}
Let $T(n,d)$ denote the worst-case minimum expected depth of any randomized LDT that reliably computes the point location problem on an $n$ point subset of $\R^d$. That is, calling the family of such randomized LDTs $\mathcal{T}_r$:
\[
T(n,d) = \max_{X \subset \R^d: |X|=n} \min_{T \in \mathcal{T}_r}\left[\ED(T)\right]
\]
\end{definition}
First, we show how point location can be reduced to matrix verification.
\begin{lemma}\label{lemma:pl-to-veri}
Point location reduces to matrix verification:
\[
T(n,d) \leq C_1d\log(d)\log(n)\log(d\log(n)) + 2V(C_2d\log(n)),
\]
for some constants $C_1,C_2$.
\end{lemma}
\begin{proof}
Given a set $X \subset \mathbb{R}^d$ of size $n$, we apply the boosting procedure used in \textsc{WeakLearn} (\Cref{lemma:weakLearner}), restricting at each step to the set of un-inferred points to learn all of $X$. Setting the failure probability of our weak learner \textsc{PartialLearn} to $p=\frac{1}{\text{poly}(d,\log(n))}$ will be more than sufficient to ensure that every weak learner (and every simulated oracle call) succeeds with probability at least $1/2$ by a union bound.

We now explain our reduction to verification. We require no more than $O(d\log(n))$ weak learners to solve the point location problem with constant probability. The inferences made by the $j$th weak learner in this process rely on a single reference point $x_j$, and a set of relatively small margin points $S_j$ of size $O(d\log(d))$. The only un-verified part of our inference is based upon the gap in margin between $x_j$ and any $v \in S_j$, and in particular on statements of the form:
\begin{align}\label{eq:verify-small-large}
\frac{|\langle v, h \rangle|}{\norm{T_j(v)}} \leq \frac{1}{c(d)}\frac{|\langle x_j, h \rangle|}{\norm{T_j(x_j)}},
\end{align}
where $T_j$ is the Barthe transform used by the $j$th weak learner, and $c(d)$ is some term dependent on dimension given by \Cref{eq:inf-small}. Our goal is to reduce verifying this set of equations to a matrix verification problem on the reference points $x_j$. In order to see this, we recall an important property of our inference from the proof of \textsc{IsoLearn} (\Cref{lemma:isotropic-oracle-learner}): we infer not only the sign of points, but also their relative margin with respect to $x_j$ up to a factor of $2$. In more detail, for each point $v$ learned by the $i$th weak learner, we also compute a quantity $C_v$ which satisfies:
\begin{align*}
\frac{1}{2}C_v\frac{|\langle x_i, h \rangle|}{\norm{T_i(x_i)}} \leq \frac{|\langle v, h \rangle|}{\norm{T_i(v)}}\leq 2C_v\frac{|\langle x_i, h \rangle|}{\norm{T_i(x_i)}}.
\end{align*}
If this holds and some point $v \in S_j$ is later learned in step $i>j$, notice that to verify \Cref{eq:verify-small-large} it is sufficient to compare just the reference points $x_i$ and $x_j$:
\begin{align}\label{eq:x-v-ver}
    2C_v\frac{\norm{T_i(v)}}{ \norm{T_j(v)}\norm{T_i(x_i)}}|\langle x_i, h \rangle| \leq \frac{1}{c(d)}\frac{|\langle x_j, h \rangle|}{\norm{T_j(x_j)}}.
\end{align}
Further, notice that this equation holds with high probability, as otherwise 
\begin{align*}
    2C_v\frac{\norm{T_i(v)}}{ \norm{T_j(v)}\norm{T_i(x_i)}}|\langle x_i, h \rangle| &> \frac{1}{c(d)}\frac{|\langle x_j, h \rangle|}{\norm{T_j(x_j)}}\\
    \implies \frac{|\langle v, h \rangle|}{\norm{T_j(v)}} &> \frac{1}{4c(d)}\frac{|\langle x_j, h \rangle|}{\norm{T_j(x_j)}},
\end{align*}
which we proved in \Cref{lemma:isotropic-oracle-learner} occurs with probability at most $p$. With these facts in mind, we induct on the number of verified weak learners, starting at the end, to prove that verifying \Cref{eq:small-ver} can be reduced to a matrix verification problem on the set of reference points $x_j$.

We begin with the final learner as our base case, each of whose $O(d\log(d))$ inequalities may each be verified in a constant number of linear queries. For the inductive step, assume then that we have verified all weak learners past step $j$. We wish to show that we can verify the inequalities for the $j$th learner by comparing $x_j$ to $x_i$, for $i>j$. Notice that the small margin points in $S_j$ must be learned in some later stage $i>j$, since they are not learned on or before step $j$ by construction. For each $i>j$, let $S_{j,i}$ denote elements in $S_j$ which are inferred in step $i$. Restating \Cref{eq:verify-small-large}, for each $v \in S_{j,i}$, we would like to verify:
\begin{align*}
\frac{|\langle v, h \rangle|}{\norm{T_j(v)}} \leq \frac{1}{c(d)}\frac{|\langle x_j, h \rangle|}{\norm{T_j(x_j)}}.
\end{align*}
Since $i>j$ by assumption, we can now apply the inductive hypothesis, that we know the relative margin of $v$ to $x_i$ up to a factor of $2$. As noted previously, it is then sufficient to verify \Cref{eq:x-v-ver}. Finally, note that it is sufficient to check only the equation with the minimum constant. In particular, let the constant $C_{ij}$ determine the smallest threshold we would like to verify over $S_{j,i}$, that is:
\[
C_{ij} = \min_{v \in S_{j,i}} \left (\frac{1}{2c(d)}\frac{ \norm{T_i(x_i)}}{ \norm{T_j(x_j)}}\frac{\norm{T_j(v)}}{C_v\norm{T_i(v)}} \right ).
\]
Then it is sufficient to verify for every pair $(x_i, x_j)$ of reference points that:
\begin{align}\label{eq:mat-ver-abs}
|\langle x_i, h \rangle| \leq C_{ij} |\langle x_j, h \rangle|
\end{align}
where all inequalities hold with probability at least $\frac{1}{2}$.
All that is left to reduce to matrix verification is to remove the absolute value signs. To do this, note that if we know $\sign(\langle x_i,h \rangle)$ and $\sign(\langle x_j,h \rangle)$, we can determine whether \Cref{eq:mat-ver-abs} holds by checking one of:
\[
\langle x_i, h \rangle \leq C_{ij} \langle x_j, h \rangle~\text{or}~\langle x_i, h \rangle \leq -C_{ij} \langle x_j, h \rangle.
\]
Modifying the sign of $C_{ij}$ appropriately for each pair, we have reduced to a matrix verification problem on the reference points. Since checking these signs adds no asymptotic complexity, this allows us to bound the expected number of queries for point location by:
\[
T(n,d) \leq \underbrace{C_1d\log(d)\log(n)\log(d\log(n))}_{\text{Reduce to Verification}} + \underbrace{V(C_2d\log(n))}_{\text{Verify}} + \underbrace{\frac{1}{2}T(n,d)}_{\text{Failed Verification}},
\]
for some constants $C_1$ and $C_2$. Collecting the $T(n,d)$ terms then gives the desired bound.
\end{proof}
While reducing from point location to matrix verification is an important step in and of itself, we have only succeeded so far in lowering the number of inequalities we need to verify to $\tilde{O}(d^2)$. The key insight for reaching a nearly linear algorithm is to notice that one can split matrix verification up into smaller point location problems in fewer dimensions, and then recurse on this process.

The intuition lies in the fact that even though there are $\tilde{O}(d^2)$ inequalities to verify, each $x_i$ has a single corresponding $x_j$ for which checking the corresponding inequality is both necessary and sufficient. In particular, recall that our inequalities have the form
\[
\langle x_i,h \rangle \leq C_{ij}\langle x_j,h \rangle.
\]
Since we need to verify that every inequality holds, checking for each $j$ only the inequality where the right hand side is minimal is both necessary and sufficient. This reduces the number of inequalities we need to check to $\tilde{O}(d)$. However, there is a slight issue: we do not know the values of $\langle x_j,h \rangle$, nor can we easily find them. Our strategy for finding these minima, a reduction to a series of low-dimensional point location problems, is shown in \Cref{alg:mv} and analyzed in the following Lemma.
\begin{algorithm}[h!]
\SetAlgoLined
\KwIn{Points $\{x_1,\ldots,x_m\}$, constraint matrix $C_{i,j}$, batch size $b$.}
\For{$0 \leq k \leq \left \lceil \frac{m}{b} \right \rceil -1$}
{
Solve the Point Location problem on $X=\{C_{ij_1}x_{j_1} - C_{ij_2}x_{j_2}\}_{1 \leq i \leq m}^{b\cdot k+1 \leq j_1,j_2 \leq b(k+1)}$.\\
\For{$1 \leq i \leq m$}
{
Using the resulting labeling, compute $j^\ast_{i,k} = \argmin\limits_{b\cdot k +1 \leq j \leq b\cdot (k+1)} \{C_{i,j}\langle x_{i,j}, h\rangle\}$.
}
}
\For{$1 \leq i \leq m$}
{
Compute $j^\ast_i = \argmin\limits_k \{C_{i,j^\ast_{i,k}}\langle x_{i,j^\ast_{i,k}}, h\rangle\}$.\\
\If{ $\langle x_i, h \rangle > C_{i,j^\ast_{i}}\langle x_{i,j^\ast_{i}}, h\rangle$}
{
\Return{False}
}
}
\Return{True}
 \caption{Matrix Verification to Point Location}
 \label{alg:mv}
\end{algorithm}

\begin{lemma}\label{lemma:ver-to-pl}
Matrix verificiation reduces to point location. For any positive integer $b<m$:
\[
V(m) \leq 2m + \frac{m^2}{b} + \frac{2m}{b} T(mb^2,b)
\]
\end{lemma}
\begin{proof}
Given the constraint matrix $C_{ij}$, our goal is to find for each row $i$ the argmin of $C_{ij}\langle x_j, h \rangle$. As soon as we have found such minimal indices $k_i$ for each row, it is sufficient to know for all $i$:
\[
\sign(\langle x_i - C_{ik_i}x_{k_i}, h \rangle),
\]
which takes a total of $m$ queries to check. Our strategy for finding each row's minimum is to divide the constraint matrix into batches of columns of size at most $b$, i.e. $C_1 = \{1,\ldots,b\},C_2=\{b+1,\ldots,2b\},\ldots$, and find the local argmin within each $C_k$ as a sub-problem. Having found these $\left \lceil \frac{m}{b} \right \rceil $ minima, we can compare them to find a global solution. Since this process must be done for each of $m$ rows, finding these globally minimal indices takes at most $m\left \lceil \frac{m}{b} \right \rceil \leq m+\frac{m^2}{b}$ queries.

They key observation is that each sub-problem may be rewritten as a point location problem in $b$ dimensions. Consider without loss of generality the sub-problem on columns $C_1$. For each $i$, our goal is to find:
\[
\underset{{j \in [b]}}{\text{argmin}} \left (C_{ij}\langle x_j, h \rangle \right ).
\]
Notice that to find this minimum, it is sufficient to have comparisons on each pair of values, that is:
\[
\sign \left(C_{ij_1}\langle x_{j_1}, h \rangle - C_{ij_2}\langle x_{j_2}, h \rangle \right)  = \sign  \left( \langle C_{ij_1}x_{j_1} - C_{ij_2} x_{j_2}, h \rangle \right),
\]
for $j_1,j_2 \in [b]$. This, however, is just a point location problem on the set of less than $mb^2$ points:
\[
X =  \{C_{ij_1}x_{j_1} - C_{ij_2}x_{j_2}\}.
\]
Noticing that $X$ lies in the at most $b$-dimensional span of $\{x_{1},\ldots, x_b\}$ then gives the desired result. 
\end{proof}
Together, \Cref{lemma:pl-to-veri} and \Cref{lemma:ver-to-pl} set up a recurrence which we can use to bound both $V(m)$ and $T(n,d)$.
\begin{corollary}\label{cor:veri-time}
The expected query complexity of a verification problem of size $m$ is at most:
\[
V(m) \leq O \left (m \cdot 2^{5\sqrt{\log(m)\log\log(m)}} \right )
\]
The expected query complexity of a point location in $d$ dimensions on sets of size $n=\text{poly}(d)$ is at most:
\[
T(n,d) \leq d \cdot 2^{O \left (\sqrt{\log(d)\log\log(d)} \right )}
\]
\end{corollary}
\begin{proof}
Plugging \Cref{lemma:pl-to-veri} into \Cref{lemma:ver-to-pl} implies a recurrence for any $b<m$:
\[
V(m) \leq \frac{m^2}{b} + C_3\log^3(m)m+4\frac{m}{b} V(C_4b\log(m))
\]
for some constants $C_3$ and $C_4$.
For simplicity, consider the form of this recurrence upon choosing $b = \frac{m}{C_4\log(m)\beta(m)}$, for $\beta(m)$ some function $\Omega(\log^2(m))$. The recurrence then reduces to:
\begin{align}\label{eq:V(n)} 
V(m) \leq C_5\log(m)\beta(m)m + C_6\log(m)\beta(m)V\left ( \frac{m}{\beta(m)} \right )
\end{align}
for some constants $C_5$ and $C_6$.
Notice that recursing to size even $\sqrt{m}$ requires at least $\Omega(\log(m)/\log(\beta(m)))$ iterations. Then at the very least any solution to this recursion must have both an $\Omega(\beta(m)m)$ term and an $\Omega(\log(m)^{\log(m)/\log(\beta(m))}m)$ term. Since the former term is monotonically increasing in $\beta(m)$, and the latter term monotonically decreasing, the optimal choice for $\beta(m)$ (asymptotically) is when these terms equalize, or in this case around:
\begin{align*}
\beta(m) = 2^{\sqrt{\log(m)\log\log(m)}}.
\end{align*}
We now solve the recurrence by induction for this particular choice of $\beta(m)$. For our base case, notice that a constant sized matrix verification problem may be solved in constant queries by brute force, giving:
\[
V(c_1) \leq c_2,
\]
for any constant $c_1$ and some corresponding constant $c_2$. For the inductive step, assume that our bound holds for $m' < m$, that is:
\[
V(m') \leq Cm'2^{5\sqrt{\log(m')\log\log(m')}}
\]
for some constant $C$. We may assume the left summand in \Cref{eq:V(n)} is smaller than the right hand term, else we are done for $m$ past some sufficiently large constant. Applying this and the inductive hypothesis, we may rewrite our bound as:
\begin{align*}
V(m) &\leq Cm2^{5\sqrt{(\log(m)-\log(\beta(m))\log(\log(m) - \log(\beta(m)))}+\log\log(m)+C_7}\\
&\leq Cm2^{5\sqrt{\log(m) \log\log(m)}}
\end{align*}
for some constant $C_7$ and sufficiently large $m$.
Plugging this bound for $V(m)$ into \Cref{lemma:pl-to-veri} immediately gives the desired bound on $T(n,d)$.
\end{proof}
Notice here that we have made the assumption $|X|=\text{poly}(d)$, despite the fact that we are interested in arbitrarily large point sets. This is due to the fact that applying \Cref{cor:veri-time} naively would result in a solution to point location problems of size $n$ with expected query complexity $(d\log(n))^{1+o(1)}$, which is just short of what we need to prove \Cref{thm:zldt}. To avoid additional terms in $n$, we will apply \Cref{cor:veri-time} to batches of weak learners rather than the process as a whole, showing first how to reduce each batch to a $\text{poly}(d)$-sized point location problem.
\begin{proof}[Proof of \Cref{thm:zldt}]
For simplicity, we will use \textsc{WeakLearn} rather than \textsc{PartialLearn}. We batch together the verification of $d^2$ learners at a time, and show that this verification can be reduced to solving a point location problem on poly($d$) points in  $d$ dimensions. 
As in \Cref{lemma:pl-to-veri}, note that we can verify our inferences by checking the relative margin of the reference points $x_i$ to their corresponding sets of small margin points $S_i$. For any $s \in S_i$ we can verify a statement like $|\langle x_i, h \rangle| \geq C_{ij}|\langle s, h \rangle|$ directly through some combination of the four linear queries: 
\[
\sign(x_i), \sign(s), \sign(\langle x_i + C_{ij}s, h \rangle), \sign(\langle x_i - C_{ij}s, h \rangle).
\]
Further, since each weak learner uses at most $\tilde{O}(d)$ reference points and $\tilde{O}(d^2)$ respective small margin points, verifying a batch of $d^2$ such learners only requires knowing the labels of at most $\tilde{O}(d^5)$ points in $d$ dimensions. This is just a point location problem on $\text{poly}(d)$ points, so by \Cref{cor:veri-time} we can reliably compute the labels in at most $d \cdot 2^{O\left(\sqrt{\log(d)\log\log(d)}\right)}$ queries in expectation.

Since our weak learner has constant coverage, we need to verify at most $c\log(n)$ (for some constant $c>0$) total learners to label all of $X$. All that is left is to analyze the expected query complexity of doing so. Since we are batching the learners together, we can treat the number of queries used by each batch as a variable and use linearity of expectation to claim that the final expected query complexity is at most the $\ceil{\frac{c\log(n)}{d^2}}$ batches times the complexity of a single batch. For our analysis, we break into two cases based on the size of $n$:
\paragraph{Case 1: $c\log(n) \geq d^2$.} In this first case, since our randomized point location solution is correct with at least constant probability, in expectation we only need to run each batch of learners plus verification a constant number of times. The $d^2$ weak learners together take $O(d^3\log^2(d))$ queries, which dwarfs the complexity of the verification process. Applying this across all $\ceil{\frac{c\log(n)}{d^2}}$ batches then gives a randomized LDT $T$ that reliably computes point location with expected depth:
\[
\ED(T) \leq O(d\log^2(d)\log(n)).
\]
Unfortunately, this analysis does not work when $c\log(n) < d^2$, since the ceiling of $\ceil{\frac{c\log(n)}{d^2}}$ may be more than just a constant times larger than $\frac{c\log(n)}{d^2}$ itself.
\paragraph{Case 2: $c\log(n) \leq d^2$.} In this case we only perform verification a single time at an expected cost of $d \cdot 2^{O(\sqrt{\log(d)\log\log(d))}}$ queries. Further, the weak learners take only $O(d\log^2(d)\log(n))$ queries in expectation, and succeed with constant probability. Thus we only need to run this process a constant number of times in expectation, giving a randomized LDT $T$ that reliably computes point location with expected depth:
\[
\ED(T) \leq O(d\log^2(d)\log(n))+d \cdot 2^{O\left(\sqrt{\log(d)\log\log(d)}\right)}.
\]
\end{proof}

\section{Binary Classifiers and Non-homogeneity}\label{sec:non-homogeneous}
In this section, we briefly discuss how to generalize our arguments to non-homogeneous hyperplanes and to binary labels. In particular, we consider labeling a set $X \subset \R^d$ with respect to a non-homogeneous halfspace $\langle \cdot, h \rangle + b$ via binary queries of the form $\sign(\langle \cdot, h\rangle + b) \in \{-,+\}$, where we assume without loss of generality that points on the hyperplane are labeled `$+$'. 

We need to address two differences: non-homogeneity, and binary rather than ternary queries. Our strategy is to homogeneously embed $X$ and $h$ into $d+1$ dimensions, sending each $x \in X$ to $x'=(x,1)$, and $h$ to $h'=(h,b)$. Notice that since $\langle x', h' \rangle = \langle x, h \rangle + b$, solving the point location problem over the embedded set $X'$ with respect to the homogeneous hyperplane $h'$ is then sufficient for our purposes. To apply our arguments, however, we need to be able to simulate queries of the form $\sign( \langle (x,\alpha) , h' \rangle )$ for $x \in \mathbb{R}^d, \alpha \in \R$. Consider the following potential query assignment:
\[
Q_{(x,\alpha)}(h') = \sign(\alpha) \cdot \sign \left ( \langle x/\alpha, h \rangle + b \right).
\]
This simulation gives the desired result in all but two circumstances: $\alpha=0$, or $\alpha < 0$ and $\langle (x,\alpha),h' \rangle = 0$. We argue that it is easy to modify our algorithm such that these cases occur with probability $0$ and/or do not adversely affect the algorithm. 

Assume for the moment that no point in $X'$ lies on the hyperplane $h'$. In this scenario, we can group our queries into two types, random labels, and random comparisons. In particular, for $x$ a random combination of points in $X'$, our queries are either of the form $\sign(\langle x, h' \rangle )$, or $\sign(\langle x - c(x,y)y, h' \rangle)$ for some $y \in \mathbb{R}^{d+1}$ and $c(x,y)$ some coefficient possibly dependent on $x$ and $y$.
Notice that because $X'$ lies off the hyperplane, random label queries will avoid both bad scenarios with probability $1$. Further, if either does occur, at worst it may cause us to restart our algorithm (with some care, neither case will cause a mistake). Comparisons are slightly more nuanced due to the fact that $c(x,y)$ depends on $x$ and $y$. Here the key fact is that none of our methods require using exactly $c(x,y)$. This means that we can (with probability $1$) fudge $c(x,y)$ slightly to avoid the $\alpha=0$ case, and that if $x-c(x,y)y$ does happen to lie on the hyperplane, it does not matter whether we receive `$+$' or `$-$'.

In general, however, we cannot assume no point in $X'$ lies on the hyperplane $(h,b)$. In our standard argument, we deal with this degenerate case simply by checking it manually, but since our queries are now binary rather than ternary, this is no longer so simple. Instead, we argue that we can reduce to the case that no point lies on the hyperplane via shifting it by an infinitesimal $\delta$ to $(h,b+\delta)$. Notice that this shifted problem has two properties: no point in $X'$ lies on the new hyperplane, and every query response with no infinitesimal part stays the same. Thus this problem is solvable by the previous argument, and indistinguishable from the case that points in $X'$ do lie on the hyperplane, which completes the argument for general $X'$ and in turn for general $X$.

\bibliographystyle{unsrtnat}  
\bibliography{references} 

\appendix

\section{Characterizing Isotropic Transformation}\label{app:iso}
In this section, we provide an exact characterization of when a pair $(X,\mu)$ can be transformed into isotropic position. First, recall the definition of isotropic position (previously $0$-isotropic position):
\begin{definition}
A pair $(X \subset \mathbb{R}^d, \mu)$ lies in isotropic position if:
\[
\forall v \in 
\R^d~:~\sum\limits_{x \in X}\mu(x) \frac{\langle x,v \rangle^2}{\norm{v}^2} = \frac{1}{d}.
\]
\end{definition}
Forster \cite{forster2002linear} proved that any uniformly weighted set $X \subset \R^d$ in general position may be transformed into isotropic position. We generalize Forster's result by showing that a weaker condition (similar to that of \Cref{cor:iso}) is both necessary and sufficient.
\begin{theorem}
Given a pair $(X,\mu)$, there exists an invertible linear transformation $T$ such that $(X_T,\mu_T)$ is in isotropic position if and only if for all $0 < k < d$, every $k$-dimensional subspace $V$ satisfies either:
\begin{enumerate}
\item $\mu(X \cap V) < \frac{k}{d}$, or
\item $\mu(X \cap V) = \frac{k}{d}$ and the remaining mass lies in a $(d-k)$-dimensional subspace.
\end{enumerate}
\end{theorem}
\begin{proof}
We begin by showing that this condition is necessary. By the same argument as \Cref{cor:iso}, there cannot exist a subspace $V$ with more than a $k/d$ fraction of $(X,\mu)$. Assume then there exists some $V$ containing exactly a $k/d$-fraction of $(X,\mu)$, but the remainder does not lie entirely in a $(d-k)$-dimensional subspace. After the application of any invertible linear transformation $T$, there will exist a $k$-dimensional subspace $T(V)$ with a
$k/d$-fraction of $(X_T,\mu_T)$. Further, the remaining mass in $X_T$ cannot be entirely orthogonal to $T(V)$, as this would imply it lies in a $(d-k)$-dimensional subspace. Since some point with non-zero measure then has non-zero projection onto $T(V)$, this forces the trace of the covariance matrix along $T(V)$ to be more than $k/d$.

Next, we show that this condition is sufficient. We  proceed by induction. The base case $d=1$ trivially holds. For the inductive step, we first show it is sufficient to consider the case where no $k$-dimensional subspace $V$ contains a $k/d$ fraction of $(X,\mu)$ for any $0<k<d$. If such a $V$ did exist, by our assumption there must exist a subspace $V'$ of dimension $d-k$ containing the remaining mass. Note that $V$ and $V'$ are complementary, as otherwise the entire measure of $X$ would
be contained in a subspace of dimension $d-1$ which is contrary to our assumptions. By the inductive hypothesis, we have isotropic transforms for $V \cap X$ and $V' \cap X$. Performing them on their appropriate subspaces and then making their images
orthogonal yields the desired isotropic transform for $X$.

We may now assume every $k$-dimensional subspace has less than $k/d$ mass. By \Cref{cor:iso}, we know that for any $\varepsilon>0$ there exists an $\varepsilon$-approximate isotropic transform $T_\varepsilon$. We can scale $T_\varepsilon$ so that its largest singular value is 1. By compactness, there must exist some limit point $T$ of the $T_\varepsilon$ as $\varepsilon \rightarrow 0$. We claim that $T$ is the appropriate isotropic transform. Perhaps the most difficult part of proving this lies in showing that $T$ is non-singular.

To start, note that by continuity, the largest singular value of $T$ must be $1$. Suppose then for the sake of contradiction that $V=\ker(T)$ is a dimension $k>0$ subspace. We claim that $V$ must contain at least a $k/d$-fraction of $(X,\mu)$. To see this, let $H=\text{Im}(T)$, and note that for any $x\in X \setminus (X \cap V)$ it must be the case that:
\[
\lim_{\varepsilon\rightarrow 0} \frac{T_{\varepsilon}(x)}{|T_{\varepsilon}(x)|} = \frac{T(x)}{|T(x)|} \in H.
\] 
Then for any $\delta>0$, there is some sufficiently small $\varepsilon$ such that for all $x \in X \setminus (X \cap V)$, $T_\varepsilon(x)/|T_\varepsilon(x)|$ is within $\delta$ of $H$. For small enough $\delta$, this will cause the trace of the covariance matrix of $T_\varepsilon(x)$ along $H$ to be arbitrarily close to $\mu(X \setminus (X \cap V))$. However, as $\varepsilon$ goes to $0$, this trace must be $\dim(H)/d$, which is a contradiction unless a $k/d$-fraction of $(X,\mu)$ lies in $V$.

Therefore, $T$ is non-singular and for every $x\in X$, we have that $T(x)/|T(x)| = \lim_{\varepsilon\rightarrow 0} T_\varepsilon(x)/|T_\varepsilon(x)|$. The covariance matrix of $T(X)$ is then just the limit of the covariance matrices of $T_\varepsilon(X)$, which approach $I/d$, making $T$ the desired isotropic transform.
\end{proof}
\end{document}

%% file: macro.tex
\newtheorem{theorem}{Theorem}[section]
\newtheorem{corollary}[theorem]{Corollary}
\newtheorem{proposition}[theorem]{Proposition}
\newtheorem{lemma}[theorem]{Lemma}
\newtheorem{definition}[theorem]{Definition}

\newtheorem{claim}[theorem]{Claim}
\newcommand{\R}{\mathbb{R}}

\newcommand{\bigo}{\mathcal{O}}

\DeclareMathOperator{\sign}{sign}
\DeclareMathOperator{\poly}{poly}

\newcommand{\E}[2]{\mathbb{E}_{{#1}}\left[{#2}\right]}
\newcommand{\q}[1]{``{#1}''}
\newcommand\inner[2]{|\langle #1, #2 \rangle|}

\newcommand{\ceil}[1]{\lceil {#1} \rceil}
\newcommand{\norm}[1]{\| {#1} \|}

\newcommand{\lm}{\mathrm{Large Margin}}
\newcommand{\mg}{\mathrm{Margin Gap}}

\newcommand{\Ebb}{\mathbb{E}}

\newcommand{\ED}{\text{ED}}
\newcommand{\MD}{\text{MD}}
\DeclareMathOperator*{\argmin}{arg\,min}